\documentclass[fleqn,11pt]{article}

\usepackage{rotating}
\usepackage{booktabs}

\usepackage{graphicx}
\usepackage{float}
\usepackage{standalone}
\usepackage{mathrsfs}
\usepackage{wrapfig}
\usepackage{float}
\usepackage{stmaryrd}
\usepackage{amsmath}
\usepackage{subfigure}
\usepackage{pgf,tikz}
\usetikzlibrary{arrows}
\usetikzlibrary{babel}
\usepackage{wrapfig}
\usetikzlibrary{calc,patterns,angles,quotes}
\usepackage{vmargin}
\usepackage{longtable}
\usepackage{subcaption}
\usepackage{pdfpages}
\usepackage{pdflscape}
\usepackage{xfrac}
\usepackage{authblk}
\usepackage[utf8]{inputenc}
\usepackage{enumitem} 
\usepackage{natbib}
\usepackage{url}
\usepackage{extarrows}

\usepackage{amssymb}
\usepackage{amsmath}
\usepackage{amsthm}
\usepackage{mathabx}
\usepackage{empheq}
\usepackage{mathtools}

\newcommand*\widefbox[1]{\fbox{\hspace{2em}#1\hspace{2em}}}

\usepackage[hang, flushmargin]{footmisc}
\usepackage[colorlinks = true,
            linkcolor = blue,
            urlcolor  = blue,
            citecolor = blue,
            anchorcolor = blue]{hyperref}

\newtheorem{theorem}{Theorem}
\newtheorem{lemma}{Lemma}
\newtheorem*{remark}{Remark}

\numberwithin{equation}{section}

\usepackage{mathrsfs}

\usepackage{array}
\newcolumntype{P}[1]{>{\centering\arraybackslash}p{#1}}

\def\dotminus{\mathbin{\ooalign{\hss\raise1ex\hbox{.}\hss\cr
  \mathsurround=0pt$-$}}}

\begin{document}

\title{A predator-prey model with age-structured role reversal }
\author[1]{Luis Carlos Suárez\thanks{lcsuarez@umd.edu}}
\author[1]{Maria Cameron\thanks{mariakc@umd.edu}}
\author[2]{William F. Fagan\thanks{bfagan@umd.edu}}
\author[1]{Doron Levy\thanks{dlevy@umd.edu}}
\affil[1]{Department of Mathematics, University of Maryland, College Park, MD 20742, USA}
\affil[2]{Department of Biology, University of Maryland, College Park, MD 20742, USA}

\maketitle

\abstract{
We propose a predator-prey model with an age-structured predator population that exhibits a functional role reversal. The structure of the predator population in our model embodies the ecological concept of an "ontogenetic niche shift," in which a species' functional role changes as it grows. This structure adds complexity to our model but increases its biological relevance. The time evolution of the age-structured predator population is motivated by the Kermack-McKendrick Renewal Equation (KMRE). Unlike KMRE, the predator population's birth and death rate functions depend on the prey population's size. We establish the existence, uniqueness, and positivity of the solutions to the proposed model's initial value problem. The dynamical properties of the proposed model are investigated via Latin Hypercube Sampling in the 15-dimensional space of its parameters. Our Linear Discriminant Analysis suggests that the most influential parameters are the maturation age of the predator and the rate of consumption of juvenile predators by the prey. We carry out a detailed study of the long-term behavior of the proposed model as a function of these two parameters. In addition, we reduce the proposed age-structured model to ordinary and delayed differential equation (ODE and DDE) models. The comparison of the long-term behavior of the ODE, DDE, and the age-structured models with matching parameter settings shows that the age structure promotes the instability of the Coexistence Equilibrium and the emergence of the Coexistence Periodic Attractor.}

%{\bf MSC}
%
%	92D25  	Population dynamics (general)
%
%    35L04  	Initial-boundary value problems for first-order hyperbolic equations
%
%    65M06  	Finite difference methods for initial value and initial-boundary value problems involving PDEs

{\bf Keywords:} age-structured, role reversal, Kermack-McKendrick renewal equation, Latin hypercube sampling, linear discriminant analysis, phase diagrams

\section{Introduction}
\subsection{An overview}
{ People have been thinking about trophic interactions between various species since ancient times. The concept of a food chain was introduced by a 10th-century Arab philosopher, al-Jahiz (\cite{Agutter2008}). The birth of modern ecology, recognizing complex interactions between species and equipped with mathematical modeling, can be attributed to the 1920s, when \cite{Elton1927} upgraded food chains to food webs and \cite{Lotka} and \cite{Volterra} proposed their predator-prey models.

While food webs can sometimes correctly capture the big picture for ecological systems, they typically oversimplify complex inter- and intraspecific interactions by ignoring an aspect of crucial importance: the variance in body size among conspecific individuals. Such variance is often associated with age differences, but other factors, such as differential foraging success, can also lead to variance in body size. \cite{WG1984} argued that changes in the body size of animals have a big influence on ecological processes. For example, juvenile predators may have very different diets from conspecific adults, and they may be eaten by other species that also compete with the prey of conspecific adults for resources. Such interactions, which biologists term ``ontogenetic niche shifts," mean that a species' functional role changes as it grows. Such shifts may negatively affect the survival and recruitment of predators, creating the so-called juvenile predator bottleneck. 

Further complexity of ecosystems is caused by role-reversal and cannibalism (\cite{Polis1988ExploitationCA}). \cite{Polis1991} pointed out four problems commonly present in catalogs of empirical food webs, making them inadequate for analysis: aggregation and undercounting of species, inadequate dietary information, lack of attention to age structure and ontogenetic niche shifts, and neglect of ecological feeding loops that occur due to role reversal and cannibalism. Polis gave several examples of ontogenetic reversals in the ecosystem of the Coachella Valley in Riverside County, California. For instance, gopher snakes eat eggs and young burrowing owls, while adult burrowing owls eat young gopher snakes. 

In fact, biological instances of role reversal and cannibalism are quite common in nature (\cite{Polis1988ExploitationCA, Polis1991}). 
%In the ecosystem of salt marshes along the US Atlantic Coast, adult killifish eat the shrimp \cite{shrimp1,shrimp2}.
% Barnacles filter and eat zooplankton, including larval echinoderms, while adult echinoderms eat barnacles~\cite{barnacles}. 
For example, in the small deciduous forest on the campus of Hokkaido University, Japan, spider mites kill the larvae of one of their key predators~(\cite{Saito1986}). In another system, cod, the top predator in the Baltic Sea, prey on herring and sprat that consume the eggs of cod~(\cite{Koster2000}), and juvenile cod between ages 0 and 2 may also lose from about a third to nearly half of its population due to cannibalism~(\cite{Neuenfeldt2000TrophodynamicCO}). In some systems, these interactions may be sublethal, such as the size-dependent aggression towards siblings that takes place among \textit{Dendrobates tinctorius} tadpoles~(\cite{Fou2022}).  

In the review paper of 2015, \cite{NAK2015} gave examples of three types of ontogenetic niche shifts---in diet, interaction type, and habitat---and emphasized the necessity of accounting for ontogenetic perspectives for understanding mechanisms underlying biodiversity and ecosystem functioning.}

\subsection{Mathematical modeling of ecosystems}
{ Numerous mathematical models of ecological systems, with various degrees of complexity, were designed to capture and analyze particular aspects of the systems' dynamics. Models with age-structured predator population based on the McKendrick equation~(\citeyear{McKendrickXLVTheRO}) were developed and studied by \cite{Cushing1982} and ~\cite{Cushing1984,
Cushing1986}. Size-structured models featuring cannibalism based on the continuity equation were proposed by \cite{Cushing1992} and \cite{FOdell}. An ODE-based predator-prey model with role reversal developed by \cite{Lehtinen2020} demonstrated the feasibility of many ecological scenarios. The sudden ecological shifts from coexistence to predator extinction may occur due to various catastrophic bifurcations, such as saddle-node, homoclinic, and subcritical Hopf. Li, Liu, and Wei~(\citeyear{LLW2022}) proposed a simple and enlightening generic ODE-based prey-predator model with role reversal. We use their model as a building block in this work. Delayed differential equations (DDEs) offer an alternative and, perhaps, simpler way to account for the age structure than the McKendrick equation~(\cite{DelayPredatorMohr}). Recently, \cite{DelayPredatorMishra} introduced a DDE-based model with role reversal. They showed that the maturation time of the predator and the prey handling time are the key factors in the dynamics of this model.

These models of complex interactions between predator and prey can be summarized as follows.
\begin{itemize}
    \item The dynamics depend on model type (ODE, DDE, PDE, etc.), model settings, and model parameters and can be complex.
    \item The predictions of various models are hard to compare to each other as the authors focus on different aspects of ecological processes.
\end{itemize}
}

\subsection{The goal and a summary of main results}
{ Here, our goal is to understand how age structure of the predator population and role reversal jointly affect the dynamics of a coupled predator-prey system.  The first objective is to develop and investigate a predator-prey model with an age-structured predator population and role reversal. The second objective is to examine how the model type affects the long-term dynamics.

Thus, we develop a mixed ordinary and partial differential equation-based model featuring an age-structured predator population and role reversal. The time evolution of the age density of the predator population is modeled by a Kermack-McKendrick-type Renewal Equation with birth- and death-rate functions depending on the prey population size. Note that the term \emph{population size} refers to the density or numbers of prey or predators rather than their physical size. The prey population, in turn, depends on the juvenile and adult predator population sizes.  This makes the predator PDE analytically unsolvable in contrast to that in (\cite{DelayPredatorMishra}). The dynamics of the prey population are governed by an ODE mimicking the prey equation in (\cite{LLW2022}). 

We analyze the proposed model using analytical and numerical tools. We prove the existence, uniqueness, and positivity of the solution to the initial value problem for the proposed model, { and the convergence of the numerical solution to it}. In addition, we obtain systems of ODEs and, separately, DDEs that approximate the dynamics of the age-structured predator population.
\color{black}
The proposed model involves 15 parameters. We use the Latin Hypercube Sampling Method (\cite{LHS1977,LHS1979,LHS1981}) to examine the long-term behavior of the model over the parameter space. The results show that the system settles to one of three attractors, which we term the {Equilibrial Coexistence Attractor}, the {Periodic Coexistence Attractor}, and the {Predator-Free Attractor}, in 22\%, 19\%, and 55\% of cases, respectively. In the remaining { approximately 4\%} of cases, the prey and predator population sizes blow up. { This blow-up is possible due to unbounded birth rates of prey and predators and a complex interplay between their population sizes involving time delay effects.} We conduct a Linear Discriminant Analysis (LDA) (a.k.a. Multiple Discriminant Analysis (MDA)) (\cite{DHS2001}) and conclude that the two most influential parameters on the qualitative long-term behavior of the system are the maturation age $\tau^{\ast}$ and the consumption rate of juvenile predators by prey $g$.

Subsequently, we fix most parameters at some default values and conduct a thorough investigation of the behavior of the proposed model on 
the two most pivotal parameters, $\tau^*$ and $g$. { We plot phase diagrams in $(\tau^*,g)$-space, denoting regions of the various attractor types, and a collection of bifurcation diagrams. We do so for two settings of the model characterized by a sharp or a gradual transition from the juvenile predator to adult. The regions of various attractor types in these settings are somewhat different.  }

We present a collection of bar plots of the age density of the predator at Equilibrial Coexistence Attractors and at a Periodic Coexistence Attractor. The age density is a monotonically decreasing and fast-decaying function for the Equilibrial Coexistence Attractors. It exhibits traveling waves with decaying amplitude during the period of the Periodic Coexistence Attractor.  In all cases, the predator population size decays to negligibly small values before reaching the chosen age cap.

Finally, we derive ODE- and DDE-based models from our age-structured model, { aiming to make the equations of all three models consistent. We} compare the phase diagrams of these three models in the $(\tau^*,g)$-space.  Evidently, the age structure of the predator population promotes oscillatory behavior. To be precise, Periodic Coexistence Attractors exist in all three models, but the sizes of the regions in the parameter space $(\tau^\ast,g)$ where the Periodic Coexistence Attractor is observed have drastically different sizes across models: the largest for the age-structured model, the smallest for the ODE model, and of an intermediate size for the DDE model. 

Our codes for simulating the age-structured model and the corresponding ODE and DDE models and plotting phase and bifurcation diagrams are available on GitHub~(\cite{mar1akc}).

The remainder of the paper is organized as follows. The age-structured model is developed in Section 2. The conditions for existence, uniqueness, and positivity are stated in Section 3 and proven in Appendix \ref{sec:proof}. The ODE and DDE models are derived in Sections \ref{sec:conn2ODE} and \ref{sec:conn2DDE} respectively. Numerical algorithms are described in Section \ref{sec:numerics}. The results of numerical studies of the age-structured model and the corresponding ODE and DDE models are presented in Section \ref{sec:results}. Our findings are discussed in Section \ref{sec:discussion}, and conclusions are drawn in Section \ref{sec:conclusion}.

\section{Model development}
\label{sec:model}
In this section, we develop a mixed ODE- and PDE-based predator-prey model with age-structured predator population and role reversal.

\subsection{A building block: Li, Liu, and Wei's ODE}
\cite{LLW2022} proposed the following ODE model for a generic predator-prey system with role reversal:
\begin{align}
         x' &= x(r-ax+sy_1-by_2), \notag\\
         y_1' &=kxy_2 - y_1(gx+D+m_1), \label{eq:LLW2022}  \\
         y_2'& = Dy_1 - m_2y_2. \notag
\end{align}
 In this model, the population of predators is subdivided into juvenile and adult predators. 
 $x$, $y_1$, and $y_2$ represent the population sizes of prey and the juvenile and adult predator respectively, while $a$, $s$, $b$, $k$, $g$, $D$, $m_1$ and $m_2$ are parameters whose biological meaning and values are listed in Table \ref{table2} -- we borrowed the selected values from (\cite{LLW2022}) except for that of $b$ that we increased from 0.4 in (\cite{LLW2022}) to $0.8$ to obtain richer dynamics.} 
 Juvenile predators are assumed to lack hunting skills and depend on adult predators for food. However, due to their assumed smaller size, juvenile predators may be eaten by the prey species, which also obtains resources from the environment and is, in turn, eaten by adult predators.

The trivial extinction equilibrium $E_0 = (0, 0,0)$ is always a saddle point with a single positive eigenvalue $r$ corresponding to the eigenvector $(1,0,0)$. The other equilibria of \eqref{eq:LLW2022}, the {Predator-Free Equilibrium} $E_1 = ( \sfrac{r}{a},0, 0)$ and the {Coexistence Equilibrium} $E^\star= (x^\star, y_1^\star, y_2^\star)$, are asymptotically stable or unstable depending on parameter values. { When $E^\star= (x^\star, y_1^\star, y_2^\star)$ is unstable,
% In addition, at certain parameter values, 
there exists the {Periodic Coexistence Attractor} around it~(\cite{LLW2022}). }

The important conditions for the existence of attractors where both species maintain population sizes bounded away from zero are $s\le g$ and $k\le b$~(\cite{LLW2022}). 
These conditions mean that the eaten population is damaged more than the eating population benefits from consuming it. 

While model \eqref{eq:LLW2022} captures that juvenile predators do not reproduce and lack hunting skills, it is still highly idealized. Most importantly, it ignores the time delay between the birth of juvenile predators and their maturation. Instead, juvenile predators become adults at the rate $D$.

\subsection{The Kermack-McKendrick Renewal Equation}
The dynamics of an age-structured population are described by the classical \emph{Kermack-McKendrick Renewal Equation} (KMRE) (\cite{PerthameTransport}, Section 3.1) 
\begin{align}
\label{eq:KMRE}
    & u_t(t,\tau)+u_\tau(t,\tau) = - \mu(\tau) u(t,\tau), \\
    &u(t,0) = \int_0^\infty B(\tau) u(t,\tau) d  \tau, \hspace{5mm} u(0,\tau)=u_0(\tau).\notag
\end{align}
Here, the independent variables $t$ and $\tau$ represent time and age, respectively;  $u(t,\tau)$ is the age density, and $\mu(\tau)$ and $B(\tau)$ are the death and birth rate functions, respectively. The age of the generation born at time $t_0$ is $\tau = t - t_0$. 
% Hence,
% \begin{align}
% \label{dtaudt}
%     \frac{d \tau}{d t} =1.
% \end{align}
% The KMRE can be understood as follows. Let $t$ be the current time. The size of the generation born within the time interval $\Delta t$ around time $t_0 < t$ is $u(t,t-t_0)\Delta t + O(\Delta t^2)$. The total time derivative of the generation size is
% \begin{align}
%     \frac{d}{dt}\left[ u(t,t-t_0) \Delta t + O(\Delta t^2)\right]& = \left[u_t(t,t-t_0) + u_\tau(t,t-t_0) \frac{d \tau}{d t} \right]\Delta t+ O(\Delta t^2) \notag\\
%     &=\left[u_t(t,t-t_0) + u_\tau(t,t-t_0)\right]\Delta t+ O(\Delta t^2). \label{eq:deriveKMRE}
% \end{align}
% Dividing \eqref{eq:deriveKMRE} by $\Delta t$ and taking the limit $ \Delta t\rightarrow 0$ we obtain the left-hand side of \eqref{eq:KMRE}. Therefore, 
KMRE says that the age density of any generation decreases at the age-dependent rate $\mu(\tau)$. 

The initial age density of the predator population,  $u_0(\tau)$, is bounded and integrable. Furthermore,  the  birth and death rate functions are bounded and satisfy the condition
\begin{align}\label{eq:condition}
    1<\int_0^\infty B(\tau)e^{-M(\tau)}d  \tau < \infty,\quad {\rm where}\quad
    M(\tau):=\int_0^\tau \mu(s) d  s.
\end{align}
These assumptions are enough to guarantee the existence and uniqueness of the solutions (\cite{PerthameTransport}). Furthermore, equation \eqref{eq:KMRE} { admits an implicit analytical solution under these assumptions, obtained} via the method of characteristics:
\begin{equation}
    \begin{cases}
        u(t,t+\tau_0) =u_0(\tau_0)e^{-\int_0^t\mu(t'+\tau_0)dt'},& \tau \equiv t + \tau_0 \ge t,\\ 
 u(t_0,t_0+\tau) =\left(\int_0^{\infty}B(\tau')u(t_0,\tau')d\tau'\right)e^{-\int_0^\tau\mu(\tau')d\tau'},& t \equiv t_0 + \tau \ge \tau,
    \end{cases}
\end{equation}
where $\tau_0$ is the age at time zero and $t_0$ is the birth time.

\subsection{An age-structured predator-prey PDE-based model}
 We will model the predator population by a renewal equation similar to KMRE \eqref{eq:KMRE}. The difference between KMRE and our renewal equation governing the predator population dynamics with the role reversal is two-fold. First, the birth and death rate functions of the predator depend on the prey population size, $x$, governed by an ODE. Second, the predator population is divided into two subpopulations of juvenile and adult predators, where the prey can eat the juvenile predators and are, in turn, eaten by the adult predators. 

An important parameter of our model is the  \emph{maturation age}
$\tau^*>0$. It is used to split the total predator population, $y(t)$, to juvenile, $y_1(t)$, and adult, $y_2(t)$, populations and to define predator's birth and death rate functions, $B(x,\tau)$ and $\mu(x,\tau)$. 
The juvenile, adult, and total predator populations are given, respectively, by 
\begin{align*}
% \label{eq:y1y2}
    y_1(t) = \int_0^{\tau*} u(t,\tau) d  \tau, \quad y_2(t) = \int_{\tau^*}^{\infty} u(t,\tau) d  \tau,\quad
    y(t)= y_1(t)+y_2(t).
\end{align*}
The ODE for the prey population, $x$, is borrowed from  (\cite{LLW2022}).
%, \eqref{eq:LLW2022}, where $y_1$ and $y_2$ are defined by \eqref{eq:y1y2}.
The predator birth and death rate functions, $B$ and $\mu$, are chosen to be of the form
\begin{align*}
    & B(x,\tau) = kx\varphi_{\{\tau \geq\tau^*\}}(\tau) + \tilde{B}(\tau)(1-e^{-\zeta x}), \\ 
    & \mu(x,\tau) = gx\varphi_{\{\tau < \tau^*\}}(\tau) + \mu_B(\tau) + \mu_Me^{-\rho x}.
\end{align*}
Here, $\zeta$ and $\rho$ are positive constants. 
The functions $\varphi_{\{\cdot\}}$ are smooth approximations of the characteristic functions of $[\tau^*,\infty)$ and $(-\infty,\tau^{\ast}]$ defined, respectively, by
\begin{align*}
    &\varphi_{{\{\tau\geq\tau^*\}}}(\tau) = \frac{1}{1+e^{-\nu(\tau-\tau^*)}},\quad 
    &\varphi_{{\{\tau<\tau^*\}}}(\tau) = \frac{1}{1+e^{-\nu(\tau^*-\tau)}}
\end{align*}
where $\nu$ is a positive constant. 
As $\nu \rightarrow \infty$, the functions $\varphi_{\{\cdot\}}$ converge pointwise to the indicator functions of the corresponding intervals. As $\nu \rightarrow 0$, $\varphi_{\{\cdot\}}$ converges pointwise to $\sfrac{1}{2}$ for all $\tau$.
{ We chose $\varphi_{\{\cdot\}}$ to be smooth to account for the gradual maturation process that is typical of many species. In nature, the transition from juveniles to adults may span a period from a few minutes to many years, depending on the species.}
The function $\tilde{B}(\tau)$ is the base birth rate supported on $[\tau^*,\infty)$ so that juvenile predators are unable to reproduce. It is nonnegative, bounded, and integrable. Importantly, the birth rate function $B(x,\tau)$ is zero at any age $\tau$ if the prey population size $x$ is zero, and it is positive at any positive $x$.

The choice of the death rate function $\mu(x,\tau)$ is motivated by the expected behavior of the predator. It is the sum of three terms describing three reasons for predators to die. The juvenile predator death rate, $gx\varphi_{\{\tau < \tau^*\}}$, is due to being eaten by prey. The base predator death rate, $\mu_B(\tau)$, accounts for deaths for age-related reasons. 
%{It is assumed to grow exponentially as $\tau\rightarrow\infty$.} 
The term $\mu_Me^{-\rho x}$ describes the death rate from hunger. 

In summary, the proposed model with the age-structured predator population and role reversal is given by
\begin{subequations}
\begin{empheq}[box=\widefbox]{align}
    x' &= x(r-ax+sy_1-by_2),\quad  x(0) = x_0, \label{eq:model-prey}\\
    u_t + u_\tau &= -\mu(x,\tau)u(t,\tau),\quad u(0,\tau) = u_0(\tau),  \label{eq:model-predator}\\
   u(t,0)& = \int_0^\infty B(x,\tau)u(t,\tau)d  \tau, \label{eq:model-BC} \\
   y_1(t) &= \int_0^{\tau*} u(t,\tau) d  \tau, \quad y_2(t) = \int_{\tau^*}^{\infty} u(t,\tau) d  \tau,\label{eq:y1y2}\\
    B(x,\tau)& = kx\varphi_{\{\tau \geq\tau^*\}}(\tau) + \tilde{B}(\tau)(1-e^{-\zeta x})\label{eq:birth-function}\\
    \mu(x,\tau)& = gx\varphi_{\{\tau < \tau^*\}}(\tau) + \mu_B(\tau) + \mu_Me^{-\rho x},\label{eq:death-function}\\
    \varphi_{{\{\tau\geq\tau^*\}}}(\tau)& = \frac{1}{1+e^{-\nu(\tau-\tau^*)}},\quad
    \varphi_{{\{\tau<\tau^*\}}}(\tau) = \frac{1}{1+e^{-\nu(\tau^*-\tau)}}. \label{eq:smooth_phi}
\end{empheq}
\end{subequations}
The initial size of the prey population $x_0$ is positive. The initial age density of the predator, $u_0(\tau)$, is nonnegative and has a compact support.

% \subsection{The ODE model obtained by averaging the parameter functions and the connection to the original ODE model}

\section{Existence, uniqueness, and positivity}
\label{sec:existence}

In this section, we will establish the existence, uniqueness, and positivity of the solution to the system \eqref{eq:model-prey}--\eqref{eq:smooth_phi}. These properties are naturally expected from a model describing population dynamics. However, the task of developing a model with these properties is nontrivial. Before settling on the model \eqref{eq:model-prey}--\eqref{eq:smooth_phi}, we designed a different model motivated by the ODE model \eqref{eq:LLW2022} from~(\cite{LLW2022}) but abandoned it because it led to a U-shaped age density of the predator. Furthermore, proving these properties for the system \eqref{eq:model-prey}--\eqref{eq:smooth_phi} turned out to be harder than we expected. Similar properties for \eqref{eq:LLW2022} directly follow from Nagumo-Brezis' Theorem (Theorem 1.2, page 26 in~(\cite{PavelTangency1999})) about flow invariance. Unfortunately, this theorem is not applicable to the system \eqref{eq:model-prey}--\eqref{eq:smooth_phi} because the flow is not of the form $z' = F(t,z)$ due to the renewal of the characteristics. The proof of existence, uniqueness, and positivity of the solution to the Kermack-McKendrick Renewal Equation in (\cite{IannelliMath1995})(Chapter VII) exploits the fact that the solution can be written out analytically in the form of integrals. This is impossible to do for our model \eqref{eq:model-prey}--\eqref{eq:smooth_phi} because of the interdependence of $x$ and $u$ as the predator's birth and death rates, $B(x,t)$ and $\mu(x,t)$, depend on $x$, and $x$ depends on $u$. As a result, our proof of the basic properties is more complicated and logically involved than the ones for \eqref{eq:LLW2022} and the KMRE.

\subsection{Characteristics of the predator equation}
 PDE \eqref{eq:model-predator} admits characteristics of the form
 \begin{equation}
     \label{eq:num:predator}
     \begin{cases}
      \frac{d}{dt}u(t,\tau_0 + t) = -\mu(x,\tau_0 + t)u(t,\tau_0 + t),& t\le \tau,~~\tau: = \tau_0 + t,\\
     \frac{d}{d\tau}u(\tau + t_0,\tau) = -\mu(x,\tau)u(\tau + t_0,\tau),& t > \tau,~~t: = \tau + t_0,
      \end{cases}
 \end{equation}
where $\tau_0$ is the initial age and $t_0$ is the birth time (see Fig. \ref{fig:age-time}).
\begin{figure}[htbp]
\begin{center}
\includegraphics[width=0.7\textwidth]{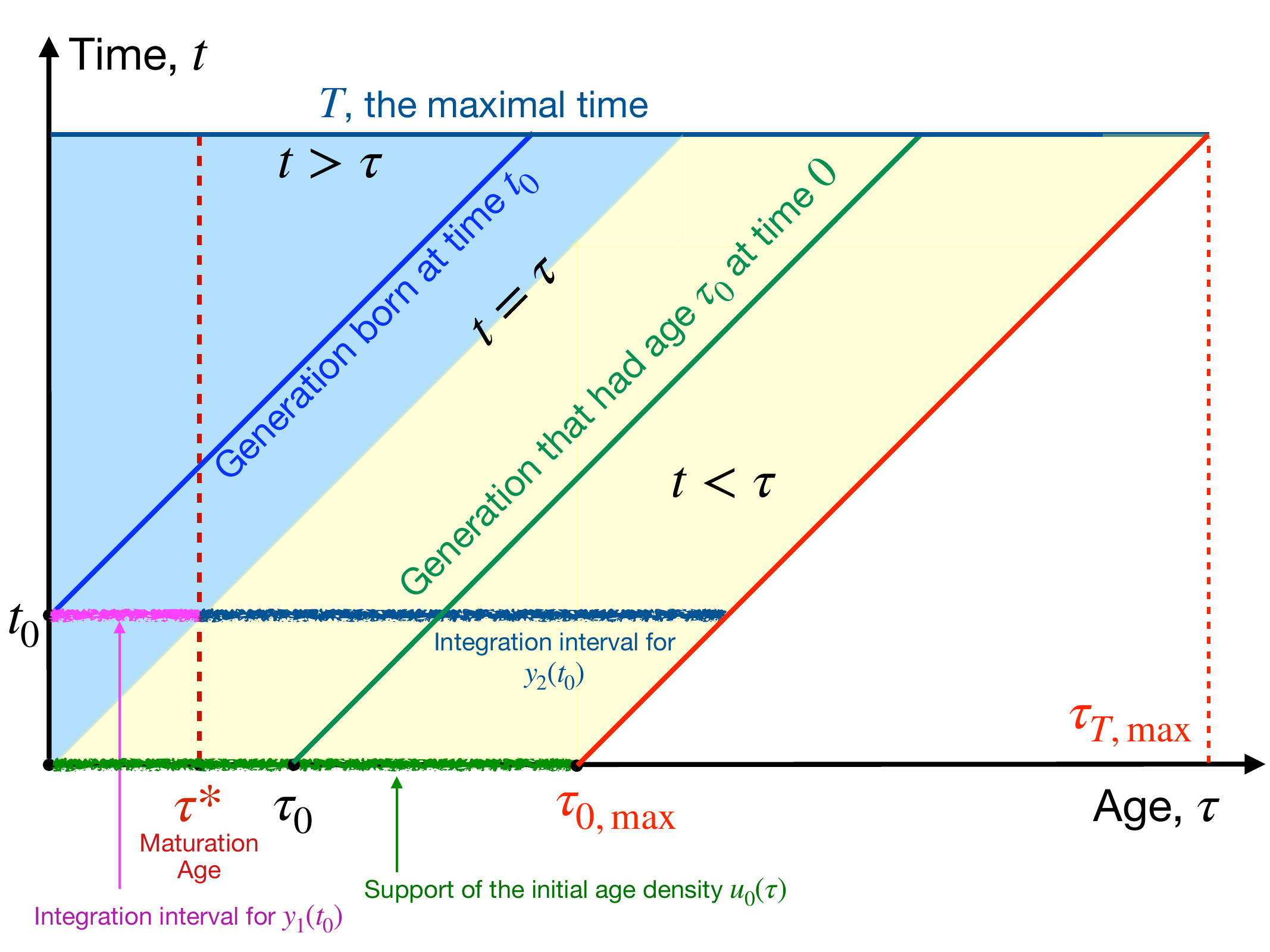}
\caption{The age-time diagram for the model {\eqref{eq:x1}--\eqref{eq:u0}. The  initial age distribution of the predator population is supported on the interval $[0,\tau_{0,\max}]$ marked with a thick green line. The time $T$ is the maximal time until which we intend to compute the solution. Therefore, the interval $[0,\tau_{T,\max}]$, where $\tau_{T,\max}: = \tau_{0,\max}+T$, is the age interval where the predator population at time $T$ is supported.}
}
\label{fig:age-time}
\end{center}
\end{figure}
With this in mind, we rewrite the system \eqref{eq:model-prey}--\eqref{eq:smooth_phi} as an infinite family of ODEs with renewal:
\begin{subequations}
\begin{empheq}[box=\widefbox]{align}
    &\frac{d}{dt} x = x\left(r - ax + s\int_0^{\tau^*} u(t,\tau)d\tau - b\int_{\tau^*}^\infty u(t,\tau) d\tau \right),\label{eq:x1}\\
    &\frac{d}{dt}u(t,t+\tau_0) = -\mu(x,t+\tau_0)u(t,t+\tau_0),\quad \tau_0 \ge 0,\label{eq:u1}\\
    &{\frac{d}{d\tau}u(\tau + t_0,\tau)  = -\mu(x,\tau)u(\tau + t_0,\tau),\quad t_0 > 0},\label{eq:u2}\\
    &u(t,0)  = \int_0^{\infty} B(x,\tau)u(t,\tau) d\tau. \label{eq:u3}\\
    &B(x,\tau) = kx\varphi_{\{\tau \geq\tau^*\}}(\tau) + \tilde{B}(\tau)(1-e^{-\zeta x})\label{eq:Bfun}\\
    &\mu(x,\tau) = gx\varphi_{\{\tau < \tau^*\}}(\tau) + \mu_B(\tau) + \mu_Me^{-\rho x},\label{eq:mufun}\\
    % &\varphi_{\nu,{\{\tau\geq\tau^*\}}}(\tau) = \frac{1}{1+e^{-\nu(\tau-\tau^*)}},\quad 
    % \varphi_{\nu,{\{\tau<\tau^*\}}}(\tau) = \frac{1}{1+e^{-\nu(\tau^*-\tau)}}. \label{eq:sphi}\\
    &x(0)  = x_0,\label{eq:x0}\\
    &u(0,\tau) = u_0(\tau) \label{eq:u0}.
\end{empheq}
\end{subequations}

\subsection{Discretization}
\label{sec:discretization}
The discretization of the model \eqref{eq:u1}--\eqref{eq:u0} serves two purposes: the numerical implementation and the proof of the existence of the solution.

The age line is discretized with step $h$. The prey ODE \eqref{eq:model-prey} and the family of predator ODEs \eqref{eq:num:predator} are discretized in time using the Forward Euler method with the same step $h$. This guarantees the propagation of generations of predators along the corresponding characteristics. The integral \eqref{eq:model-BC} for the number of newborn predators and the integrals for the numbers of the juvenile and adult predators \eqref{eq:y1y2} are calculated using the trapezoidal rule. We aim to compute a numerical solution till the maximal time $T$. Since the initial age distribution has a compact support $[0,\tau_{0,\max}]$, the maximal predator age at time $T$ is $\tau_{T,\max}:=\tau_{0,\max} + T$.
We assume that the maturation age {$\tau^*$}, the maximal initial age $\tau_{0,\max}$, and the maximal time $T$ are $h$-commensurable, i.e., $\tau^\ast = N_1h$, $\tau_{0,\max} = N_0h$, and $T = Nh$ for positive integers $N$, $N_0$, and $N_1$. Then the maximal age at time $T$ is $\tau_{T,\max} = N_2h$ where $N_2 = N_0 + N$.

 The numerical solution to the prey equation at time $nh$ is denoted by $X[n]$. The numerical solution to the predator equation at time $nh$ and age $kh$ is denoted by $U[n,k]$.
The resulting numerical scheme is:
\begin{align}
X[n+1] & = X[n]\left(1 + h\left(r - a X[n] + sY_1[n] - bY_2[n]\right)\right),\label{eq:Xnum} \\
U[n+1,k] &= U[n,k-1]\left(1 - h\mu(X[n],h(k-1)\right),\quad k=1,\ldots,N_2,\label{eq:Unum}\\
U[n+1,0] & =h\Bigg(\frac{1}{2}B(X[n],0)U[n,0] + \frac{1}{2}B(X[n],\tau_{T,\max})U[n,N_2]  \notag \\
& +\sum_{k=1}^{N_2-1} B(X[n],hk)U[n,k])\Bigg),\label{eq:birth_num}\\
Y_1[n+1] &= h\left(\frac{1}{2}U[n,0] + \frac{1}{2}U[n,N_1] +\sum_{k=1}^{N_1-1} U[n,k])\right),\label{eq:Y1num}\\
Y_2[n+1]& = h\left(\frac{1}{2}U[n,N_1] + \frac{1}{2}U[n,N_2] +\sum_{k=N_1+1}^{N_2-1} U[n,k])\right).\label{eq:Y2num}
\end{align}
The initial conditions for this scheme are
\begin{align}
    X[0] &= x_0,\label{eq:Xnum_init}\\
    U[0,k] & = u_0(hk),\quad k = 0,\ldots,N_2,\label{eq:Unum_init}\\
    Y_1[0] & = h\left(\frac{1}{2}u_0(0) + \frac{1}{2}u_0(hN_1) +\sum_{k=1}^{N_1-1} u_0(hk)\right),\label{eq:Y1num_init}\\
Y_2[0]& = h\left(\frac{1}{2}u_0(hN_1) + \frac{1}{2}u_0(hN_2) +\sum_{k=N_1+1}^{N_2-1} u_0(hk)\right).\label{eq:Y2num_init}
\end{align}

\subsection{The main theorem}
\label{sec:main_theorem}
The goal of this section is to establish the existence, uniqueness, and positivity of the solution to the initial-value problem \eqref{eq:x1}--\eqref{eq:u0}.

We define a space $\mathcal{X}$ for pairs $(x,u)$,
\begin{equation}
    \label{eq:Xdef}
    \mathcal{X} = \{(x,u)~|~x\in\mathbb{R},~u\in L_1([0,\tau_{T\max}])\},
\end{equation}
with the norm
\begin{equation}
    \label{eq:normX_main}
    \|(x,u)\| = |x| + \int_{0}^{\tau_{T,\max}} |u|d \tau.
\end{equation}
Hence, $\mathcal{X}$ with the norm \eqref{eq:normX} is Banach because it is a direct product of Banach spaces $\mathbb{R}$ and $L_1([0,\tau_{T,\max}])$.

\begin{theorem}
\label{thm1}
    Let the initial prey population size $x_0$ be positive. Let the initial predator age density $u_0(x)$ be nonnegative, piecewise smooth, and compactly supported on $[0,\tau_{0,\max}]$. Let the constants $r$, $a$, $s$, $b$, and $\mu_M$ be positive and the constants $k$ and $g$ be nonnegative. Furthermore, let the function $\tilde{B}(\tau)$ be continuous, positive, and bounded, and the function $\mu_B(\tau)$ be continuous, positive, and bounded at any finite $\tau$. Let $0\le t\le T$. Then there exists a time interval $[0,T^{\ast}]$, $T^*\le T$,  where the solution to \eqref{eq:x1}--\eqref{eq:u0} exists and is unique and positive. 
    
    { Moreover, the numerical solution by the scheme \eqref{eq:Xnum}--\eqref{eq:Y2num_init} converges to the exact solution to system \eqref{eq:x1}--\eqref{eq:u0} as the time and age step $h\rightarrow 0$, and the numerical error of the age distribution decays as $O(h)$ at any fixed time $t$ in the sense of the norm defined in \eqref{eq:normX_main}.}
\end{theorem}

The proof of Theorem \ref{thm1} is long and tedious. 
We sketch it now and elaborate it in Appendix \ref{sec:proof}.

The proof of the existence of the solution relies on the convergence of the numerical solution to a limit solution as the step $h$ tends to zero. 
% To carry out this construction, we devise a space $\mathcal{X}$ for pairs $(x,u)$,
% \begin{equation}
%     \label{eq:Xdef}
%     \mathcal{X} = \{(x,u)~|~x\in\mathbb{R},~u\in L_1([0,\tau_{T\max}])\},
% \end{equation}
% with the norm
% \begin{equation}
%     \label{eq:normX}
%     \|(x,u)\| = |x| + \int_{0}^{\tau_{T,\max}} |u|d \tau.
% \end{equation}
% Hence, $\mathcal{X}$ with the norm \eqref{eq:normX} is Banach because it is a direct product of Banach spaces $\mathbb{R}$ and $L_1([0,\tau_{T,\max}])$.
We show that the right-hand side of \eqref{eq:x1}--\eqref{eq:u3} is Lipschitz with respect to the norm \eqref{eq:normX} in a bounded region 
\begin{equation}
    \label{eq:Omega}
    \Omega_R = \{(x,u)\in \mathcal{X}~|~\|(x,u)\|\le R\},
\end{equation}
where $R$ is a positive constant. Next, we fix step $h$ and compute the numerical solution $(X_h,U_h)$ till either time $T^*$, which is the minimum of $T$, and the exit time from the region $\Omega_R$. 
%Let  $t = nh$ be any discrete time not greater than $T^*$. 
In addition, we compute a sequence of numerical solutions $(X_{h2^{-p}},U_{h2^{-p}})$ till time $t = nh\le T^*$ with steps $h2^{-p}$, $p = 1,2,\ldots$ and prove that
\begin{equation}
    \label{eq:Oh}
    \|(X_{h2^{-p}}[2^pn],U_{h2^{-p}}[2^pn,\cdot]) - (X_{h}[n],U_{h}[n,\cdot])\| = O(h),\quad p= 1,2,\ldots.
\end{equation}
Hence, this sequence of numerical solutions at time $t = nh$ is Cauchy. Therefore, it converges to an element $(x(t),u(t))\in\mathcal{X}$. 

The uniqueness follows from the Lipschitz continuity of the right-hand side of \eqref{eq:x1}--\eqref{eq:u3} and Gr\"{o}nwall's inequality. 

The proof of positivity is conducted in four stages. First, we prove that $x(t)$ is positive on $[0,T^*]$. Second, we argue that $u(t,\tau)$ is positive for $t<\tau$. Third, we show that $u(t,0)$ is positive. Finally, we conclude that $u(t,\tau)$ for $t > \tau$ is positive.

\begin{remark}
    The $x$-component and $u$ along the characteristics are continuously differentiable as immediately follows from time stepping. The function $u(t,\tau)$ has a discontinuity along the characteristic $t = \tau$ unless 
\begin{equation*}
    u_0(0) = \int_0^\infty B(x(0),\tau)u_0(\tau)d\tau.
\end{equation*}
If $u_0(\tau)$ has discontinuities at $\tau_1,\ldots,\tau_S$, $u(t,\tau)$ also has discontinuities propagating along the characteristics $t = \tau - \tau_j$, $1\le j\le S$. Therefore, the solution is well-defined as the solution to \eqref{eq:x1}--\eqref{eq:u0} everywhere except for the lines $t = \tau$ and $t=\tau - \tau_j$, $1\le j\le S$.
\end{remark}

\section{Reduction to an ODE model}
\label{sec:conn2ODE}
The goal of this section is to reduce the model \eqref{eq:model-prey}--\eqref{eq:smooth_phi} to an ODE model and compare the resulting ODE with \eqref{eq:LLW2022}~(\cite{LLW2022}). 

ODEs governing the population sizes of juvenile and adult predators are obtained by integrating
PDE \eqref{eq:model-predator} in  age $\tau$ over the intervals $[0,\tau^\ast)$ and $[\tau^\ast,\infty)$ respectively: \begin{align}
  &\frac{d}{dt}\int_0^{\tau^\ast} u(t,\tau)d  \tau + u(t,\tau^\ast) - u(t,0) = -\int_0^{\tau^\ast}\mu(x,\tau)u(t,\tau)d  \tau\label{eq:y11}, \\ 
    & \frac{d}{dt}\int_{\tau^\ast}^\infty u(t,\tau)d  \tau - u(t,\tau^\ast) = -\int_{\tau^\ast}^\infty\mu(x,\tau)u(t,\tau)d  \tau\label{eq:y21}.
\end{align}
The integrals in the right-hand sides of \eqref{eq:y11} and \eqref{eq:y21} are expanded using the definition \eqref{eq:death-function} of the death rate function. In the calculations below, we replace $\varphi_{\{\tau < \tau^*\}}(\tau)$ with the indicator function of $[0,\tau^*)$ {for the sake of simplicity}:
\begin{align}
    \int_0^{\tau^\ast}\mu(x,\tau)u(t,\tau)d  \tau &= gxy_1 + \int_0^{\tau^*}\mu_B(\tau) u(t,\tau) d  \tau + \mu_M e^{-\rho x} y_1,\label{eq:mu-1}\\
    \int_{\tau^\ast}^\infty \mu(x,\tau)u(t,\tau)d  \tau & = \int_{\tau^*}^\infty \mu_B(\tau) u(t,\tau) d  \tau + \mu_M e^{-\rho x} y_2. \label{eq:mu-2}
\end{align}
The expression for $u(t,0)$ is obtained using the definition of the birth rate function \eqref{eq:birth-function} and the assumption that $\tilde{B}(\tau)$ is zero on $[0,\tau^*)$:
\begin{align}
\label{eq:ut0}
    u(t,0) = \int_0^{\infty}B(x,\tau)u(t,\tau)d  \tau =  kxy_2 + (1-e^{-\zeta x})\int_{\tau^\ast}^\infty\tilde{B}(\tau)u(t,\tau)d \tau. 
\end{align}

Recalling the definitions of $y_1$ and $y_2$,  \eqref{eq:y1y2}, and using \eqref{eq:y11}--\eqref{eq:ut0} we obtain the following ODEs for the juvenile, $y_1$, and adult, $y_2$, predator population sizes
\begin{align}
    & y_1' + u(t,\tau^\ast) - kxy_2 - (1-e^{-\zeta x})\int_{\tau^\ast}^\infty\tilde{B}(\tau)u(t,\tau)d \tau \notag \\
     = &- gxy_1 - \int_0^{\tau^*}\mu_B(\tau) u(t,\tau) d  \tau - \mu_M e^{-\rho x} y_1,\label{eq:y12}\\
    &y_2' - u(t,\tau^*) = -\int_{\tau^*}^\infty \mu_B(\tau) u(t,\tau) d  \tau - \mu_M e^{-\rho x} y_2. \label{eq:y22}
\end{align}
%To eliminate integrals in \eqref{eq:y12} and \eqref{eq:y22}, we replace $\mu_B(\tau)$ and $\tilde{B}(\tau)$ with constants.
The term {$u(t,\tau^*)$} is the number of juvenile predators becoming adults per time unit. Therefore, we define the transition rate from juvenile to adult predator as
   { 
\begin{align}
\label{eq:D_ode}
D = \frac{\bar{u}(\tau^*)}{y_1},
\end{align}
where $\bar{u}(\tau^*)$ is the value of $u(t,\tau^*)$ at the equilibrium.}
To eliminate integrals in \eqref{eq:y12} and \eqref{eq:y22}, we introduce age-averaged birth rate $b_2$ for adult predators and age-averaged death rates $m_1$ and $m_2$ for juvenile and adult predators respectively:
\begin{align}
     b_2 &:= \frac{\int_{\tau^*}^\infty \tilde{B}(\tau)u(t,\tau)d  \tau}{\int_{\tau^*}^\infty u(t,\tau) d  \tau}, \label{eq:b2_ode}\\
     m_1 &:= \frac{\int_0^{\tau^*} \mu_B(\tau)u(t,\tau)d  \tau}{\int_0^{\tau^*} u(t,\tau) d  \tau}, \label{eq:m1_ode} \\
     m_2 &:= \frac{\int_{\tau^*}^\infty \mu_B(\tau)u(t,\tau)d  \tau}{\int_{\tau^*}^\infty u(t,\tau) d  \tau}. \label{eq:m2_ode}
\end{align}
The resulting ODE system becomes
\begin{subequations}
\begin{empheq}[box=\widefbox]{align}
    & x' = x(r - ax + sy_1 - by_2), \label{eq:Luis_prey} \\ 
    & y_1'= kxy_2 + (1-e^{-\zeta x})b_2y_2 - gxy_1 - m_1y_1 - \mu_M e^{-\rho x}y_1 - Dy_1, \label{eq:Luis_y1}\\
    & y_2' = Dy_1 - m_2y_2 - \mu_M e^{-\rho x} y_2. \label{eq:Luis_y2}
\end{empheq}
\end{subequations}

ODE model \eqref{eq:Luis_prey}--\eqref{eq:Luis_y2} exposes to connection with ODE \eqref{eq:LLW2022}~ \cite{LLW2022}. To obtain  \eqref{eq:LLW2022} from \eqref{eq:Luis_prey}--\eqref{eq:Luis_y2}, one needs to remove the natural birth rate term proportional to $b_2y_2$ and the hunger death rate terms proportional to $\mu_M e^{-\rho x}$.  In turn, to obtain the ODE system \eqref{eq:Luis_prey}--\eqref{eq:Luis_y2} from the proposed model \eqref{eq:model-prey}--\eqref{eq:smooth_phi}, we replaced the birth and death rate functions with their age averages and replaced the smoothed indicator functions with the sharp ones.

\section{Reduction to a DDE model}
\label{sec:conn2DDE}
We also can reduce the model \eqref{eq:model-prey}--\eqref{eq:smooth_phi} to a delayed differential equation (DDE). We proceed as in Section \ref{sec:conn2ODE} except for a different approximation of the term $u(t,\tau^*)$ in \eqref{eq:y12}--\eqref{eq:y22}. 
% Therefore, our starting point for the derivation of a DDE model is 
% \begin{align}
%     & x' = x(r - ax + sy_1 - by_2), \label{eq:xDDE0} \\ 
%     & y_1' =  kxy_2 + (1-e^{-\zeta x})b_2y_2 - gxy_1 - m_1y_1 - \mu_M e^{-\rho x}y_1 -u(t,\tau^*), \label{eq:y1DDE0}\\
%     & y_2' = - m_2y_2 - \mu_M e^{-\rho x} y_2 + u(t,\tau^*),\label{eq:y2DDE0}
% \end{align}
% where the constants $b_2$, $m_1$ and $m_2$ are defined by \eqref{eq:b2_ode}, \eqref{eq:m1_ode}, and \eqref{eq:m2_ode} respectively.
Since the death rate $\mu(x,\tau)$ depends on $x$ and $x$ depends on $u$, we cannot integrate the predator PDE \eqref{eq:model-predator} along the characteristics $\tau = t+\tau_0$ and $\tau = t - t_0$ explicitly. Note that the death rates in (\cite{DelayPredatorMishra}) and (\cite{DelayPredatorMohr}) were assumed independent of the prey population size -- that's why the integration was readily done in these works.
Instead, we will apply the trapezoidal rule in $x$ to obtain $u(t,\tau^*)$ { by integrating along characteristics. 

Let $t\ge\tau^* $. Then the characteristic along which we will integrate can be parametrized as $t - \tau^* + s$, $0\le s \le \tau^*$, where $t-\tau^*$ is the birth time of the corresponding generation.}   Dividing {\eqref{eq:u2}} by $u(t,\tau)$ and replacing the smooth indicator function with the sharp one we get:
{
\begin{align}
\label{eq:predDDE}
    \frac{d}{ds}\log u(t-\tau^*+s,s) = -gx(t-\tau^*+s) -\mu_B(s) - \mu_Me^{-\rho x(t-\tau^*+s)}.
\end{align}
Integrating \eqref{eq:predDDE} from $0$ to $\tau^*$ in $s$  using the trapezoidal rule for $x$ we obtain:
\begin{align}
    u(t,\tau^*) &= u(t-\tau^*,0)\exp\left(-\frac{g\tau^*}{2}[x(t-\tau^*) + x(t)] - M_B \right.\notag\\
    &-\left. \frac{\mu_M\tau^*}{2}\left(e^{-\rho x(t-\tau^*)} + e^{-\rho x(t)}\right)\right).
\end{align}
}
The initial density $u(t-\tau^*,0)$ is found as in \eqref{eq:ut0} except for $t$ is replaced with $t-\tau^*$. The symbol $M_B$ denotes the total death rate for juvenile predators:
\begin{equation}
    \label{eq:MB}
    M_B: = \int_0^{\tau}\mu_B(\tau)d\tau.
\end{equation}

{
We obtain $u(t,\tau^*)$ in the case $t < \tau^*$ in a similar manner.}
%, we integrate \eqref{eq:predDDE} starting from $t = 0$. 

Thus, the DDE model is given by
\begin{subequations}
\begin{empheq}[box=\widefbox]{align}
    & x' = x(r - ax + sy_1 - by_2), \label{eq:xDDE0} \\ 
    & y_1'= kxy_2 + (1-e^{-\zeta x})b_2y_2 - gxy_1 - m_1y_1 - \mu_M e^{-\rho x}y_1 - u(t,\tau^*), \label{eq:y1DDE0}\\
    & y_2' =  - m_2y_2 - \mu_M e^{-\rho x} y_2 + u(t,\tau^*),\label{eq:y2DDE0}\\
& \text{\bf if} ~t\ge \tau^{\ast}:\notag \\
 & u(t,\tau^*) = \left(kx(t-\tau^*)y_2(t-\tau^*) + \left(1-e^{-\zeta x(t-\tau^*)}\right)b_2y_2(t-\tau^*)\right)\notag \\
       & \times\exp\left(-\frac{{g}\tau^*}{2}[x(t-\tau^*) + x(t)] - M_B \right.\notag\\
       & \left.- \frac{\mu_M\tau^*}{2}\left(e^{-\rho x(t-\tau^*)} + e^{-\rho x(t)}\right)\right), \label{eq:uttau*0}\\
& \text{\bf if} ~t < \tau^{\ast}:\notag \\  
&  u(t,\tau^*) =  u_0(\tau^*-t)\exp\left(-\frac{{g}\tau^*}{2}[x(\tau^*-t) + x(t)] - M_B \right.\notag\\
    &-\left. \frac{\mu_M\tau^*}{2}\left(e^{-\rho x(\tau^*-t)} + e^{-\rho x(t)}\right)\right),\label{eq:uttau*1}    
\end{empheq}
\end{subequations}
where the constants $b_2$, $m_1$ and $m_2$ are defined by \eqref{eq:b2_ode}, \eqref{eq:m1_ode}, and \eqref{eq:m2_ode} respectively.

{ 
In contrast to ODE model \eqref{eq:Luis_prey}--\eqref{eq:Luis_y2}, the DDE model \eqref{eq:xDDE0}--\eqref{eq:uttau*1} explicitly contains the maturation age parameter $\tau^*$. Furthermore, this parameter enters the system of equations for the equilibrium of \eqref{eq:xDDE0}--\eqref{eq:uttau*1}, where we set the right-hand sides to zero and $x$, $y_1$, and $y_2$ to their equilibrium values: note the argument of the exponent in the expression for $u(t,\tau^*)$ for $t\ge \tau^*$. Therefore, the equilibria of the ODE and DDE models will not be exactly the same. 
}

\section{Numerics}
\label{sec:numerics}
We discretize the proposed model in the form \eqref{eq:x1}--\eqref{eq:u0} as described in Section
\ref{sec:discretization}. The resulting system is \eqref{eq:Xnum}--\eqref{eq:Y2num} with initial conditions \eqref{eq:Unum_init}--\eqref{eq:Y2num_init}. The integration of the ODE model \eqref{eq:Luis_prey}--\eqref{eq:Luis_y2} and the DDE model \eqref{eq:xDDE0}--\eqref{eq:uttau*1} have been done using Matlab's {\tt ode45} and {\tt dde23} respectively. Our codes are available on GitHub~(\cite{mar1akc}).

\subsection{Settings}
For the sake of computational convenience, we introduce the maximal predator lifespan $L = 30$. This is equivalent to having the base death rate function $\tilde{B}(\tau)$ equal to infinity for $\tau > L$. 
In Section \ref{sec:model}, we did not specify the base birth rate function $\tilde{B}(\tau)$ and the base death rate function $\mu_B(\tau)$ to keep the model more general. Now, we define these functions as
\begin{align}
      \tilde{B}(\tau) &=  \begin{cases}0,&\tau <\tau^*\\
          b_p(e^{-b_{ep}(\tau-\tau^*)}+1),&\tau\ge \tau^*
                \end{cases},\notag \\
   \mu_B(\tau) &= d_p e^{d_{ep}(\tau - L)}.
\end{align}
{This lifespan $L = 30$ is large enough so that the predator's age density is negligibly small at $\tau = L$. }

The initial condition for the prey is set to 
\begin{equation}
    x(0) = 0.5.
    \label{eq:x_init}
\end{equation} 
The initial age density for the predator is chosen to be
\begin{align}
    u_0(\tau) = \begin{cases}
        0.1,&\tau \in [0,\tau^*)\\
        0.05,&\tau\in[\tau^*,L]\\
        0,&\text{otherwise}
    \end{cases}.
    \label{eq:u_init}
\end{align}
These settings leave us with 15 parameters { whose biological meaning and values are listed} in Table \ref{table2}. {The selected parameter values, when applicable, are borrowed from ~(\cite{LLW2022}) and are not attached to any particular predator-prey system. The value of $b$, the consumption rate of prey by adult predators, is increased from 0.4 to 0.8. The ranges for the parameters are chosen around those selected values to keep the order of magnitude of the selected value. The exception is the smoothness parameter of the indicator function $\nu$. When $\nu = 100$, the indicator function is almost the Heaviside function. When $\nu = 1$, it changes gradually. For example, {$\varphi_{\{\tau\ge 1\}}(0)$} takes values of approximately $0.27$ and $3.7\cdot10^{-44}$ for $\nu = 1$ and 100, respectively.}

\begin{sidewaystable}
\caption{Parameters of the proposed model. { The values of $r$, $a$, $k$, and $s$ are borrowed from~(\cite{LLW2022}).}}\label{table2}
\begin{tabular*}{\textheight}{@{\extracolsep\fill}cccc}
\toprule%
%& \multicolumn{3}{@{}c@{}}{Element 1\footnotemark[1]}& \multicolumn{3}{@{}c@{}}{Element\footnotemark[2]} \\\cmidrule{2-4}\cmidrule{5-7}%
Parameter & Biological meaning	&Selected Value & Range \\
\midrule
$\tau^* $ & Maturation age of predator &  1 & [0,2] \\
 $g$ & Consumption rate of juvenile predator by prey& 0.2 & [0,1]   \\
   $\nu $ & Smoothness of the indicator function & 1 and 100 & [1,100]    \\
  $r$ & Intrinsic growth rate of prey & 0.4 & [0.1,0.6] \\ 
 $a$ & Intraspecific competition rate of prey & 0.01 &[0.005,0.05]\\ 
 $k$ & Reproduction rate of predator & 0.3 &[0.1,1] \\
 $b$ & Consumption rate of prey by adult predator & 0.8 & [0.1,1] \\ 
 $s$ & Growth rate of prey due to eating juvenile predators& 0.2 &  [0,1,1]  \\
$\zeta$ & Birth rate, exponential factor & 10 & [5,20]  \\
$\mu_M$ & Prefactor for the hunger death rate of predator & 1 & [0.5,5]  \\
$\rho$ & Exponential factor of the hunger death rate of predator & 5 & [3,7] \\
$d_p$ & Prefactor of the base death rate of predator & 0.4 & [0.1,1] \\
 $b_p$ & Prefactor of the base birth rate of predator  & 0.05 & [0.03,0.1]  \\
$b_{ep}$ & Exponential factor of the base birth rate of predator  & 0.1 & [0.05,0.15]  \\
$d_{ep}$ & Exponential factor of the base death rate of predator & 0.1 & [0.05,0.15]  \\
\bottomrule
\end{tabular*}
\end{sidewaystable}

\subsection{Finding stable and unstable equilibria}
\label{sec:find_equilib}
Stable and unstable equilibria of \eqref{eq:Xnum}--\eqref{eq:Y2num} are found using Newton's method. We define a vector-function $F(X^*,U^*)$, $F:\mathbb{R}^{N_2 + 2}\rightarrow \mathbb{R}^{N_2 + 2}$ by subtracting the right-hand sides of \eqref{eq:Xnum} and \eqref{eq:Unum} from their left-hand sides and replacing $X[\cdot]$ with $X^*$ and $U[\cdot,k]$ with $U^*[k]$. Then, we set $F(X^*,U^*) = 0$ and solve this equation by Newton's iteration. The $(N_2 + 2)\times(N_2 + 2)$ Jacobian matrix $J$  is approximated by forward differences
\begin{align}
    J_{:,0} &= \frac{1}{\epsilon} \left(F(X^* + \epsilon, U^*) - F(X^*,U^*)\right), \label{eq:J1}\\
    J_{:,k+1} &= \frac{1}{\epsilon} \left(F(X^*, U^* + \mathbf{e}_k\epsilon  U^*[k]) - F(X^*,U^*)\right), \quad k= 0,\ldots, N_2\label{eq:Jk}
\end{align}
where  $\epsilon = 10^{-6}$, ``:" denotes all entries from $0$ to $N_2+1$, and $\mathbf{e}_k$ is the standard $k$th basis vector in $\mathbb{R}^{N_2+1}$.
The warm start is obtained by a long integration of \eqref{eq:Xnum}--\eqref{eq:Y2num} at a set of parameter values where the equilibrium is stable and then by a continuation along a parameter under investigation. The stability of the equilibrium was assessed by computing the eigenvalues of the Jacobian matrix $J$ and checking that all of them had negative real parts.

Equilibria of the ODE and DDE systems were found using Matlab's {\tt lsqnonlin}, a nonlinear least-squares solver. The stability of the equilibria of the ODE model was checked by computing the eigenvalues of the Jacobian. The stability of the equilibria of the DDE model was checked by perturbing the equilibrium and examining where the system settled after a long time of integration.

\subsection{Finding limit cycles}
\label{sec:find_limit_cycles}
Periodic coexistence attractors exist at parameter values where the coexistence equilibrium is unstable. Periodic attractors, or limit cycles, are fixed points of Poincar\'e maps on appropriately chosen Poincar\'e sections. As a Poincar\'e section for the discretized age-structured model \eqref{eq:Xnum}--\eqref{eq:Y2num}, we choose the hyperplane in $\mathbb{R}^{N_2+2}$ $X = X^*$, where $X^*$ is the first component of the unstable coexistence equilibrium. The Poincar\'e map $G: \mathbb{R}^{N_2+1}\rightarrow\mathbb{R}^{N_2+1}$ maps the point where a trajectory of \eqref{eq:Xnum}--\eqref{eq:Y2num} crosses the Poincar\'e section to the next crossing in the same direction. The point of intersection is found using time integration of \eqref{eq:Xnum}--\eqref{eq:Y2num} and linear interpolation. The fixed points of the Poincar\'e map are found using the Levenberg-Marquardt nonlinear solver (\cite{NoceWrig06}, Section 10.3). The objective function for the Levenberg-Marquardt method is defined as $f(U) = \tfrac{1}{2}\| U - G(U)\|_2^2$. The Jacobian of $U- G(U)$ is found using forward differences similar to \eqref{eq:Jk}. The Levenberg-Marquardt method reliably found limit cycles, while Newton's method failed to do so in our settings because it required warmer starts than it was practical for us to provide.  

The limit cycles for the ODE system were found similarly. Long-time integration found the limit cycles for the DDE system.

%%%%%%%%%%%%%%%%%%%%%%%%%%%%%%%%%%%%%%%%%%%%

\section{Results}
\label{sec:results}
We subjected the age-structured model \eqref{eq:model-predator}--\eqref{eq:smooth_phi} to a detailed numerical investigation. First, we performed the Latin Hypercube Sampling (\cite{LHS1977,LHS1979,LHS1981}) in the 15-dimensional parameter space. Then we conducted the Linear Discriminant Analysis (LDA) (\cite{DHS2001}) and found which parameters affect the type of attractor the most. Next, we carried out a thorough investigation of the dynamical behavior of the age-structured model \eqref{eq:model-predator}--\eqref{eq:smooth_phi} depending on the maturation age $\tau^*$ and the consumption rate of the juvenile predator by the prey $g$, the two most important parameters. {Finally, we repeated the investigation of the dynamical behavior of the corresponding ODE,  \eqref{eq:Luis_prey}--\eqref{eq:Luis_y2},  and DDE, \eqref{eq:xDDE0}--\eqref{eq:uttau*1}, models in the same region of the $(\tau^*,g)$-plane and compared the calculated phase diagrams for all three models.}

\subsection{Latin Hypercube Sampling}
\label{sec:LHS}

The Latin Hypercube Sampling aims to examine the types of dynamical behavior that may occur in the system depending on its parameters.
As we have shown in Section \ref{sec:main_theorem}, the solution to \eqref{eq:model-predator}--\eqref{eq:smooth_phi} exists, is unique, and is positive. 

Latin Hypercube Samples were generated in 15-dimensional space formed by the direct product of the intervals specified in Table \ref{table2}. The initial condition for each parameter set was the same as in \eqref{eq:x_init} and \eqref{eq:u_init}. At each parameter set, the numerical integration continued either till time $T_{\max} = 500$  or till any of the solution components exceeded 1000.  In the latter case, we declared a blow-up. The choice of the time step for Latin Hypercube is a trade-off between reliability and runtime. We monitored the positivity of the solution components throughout the run. Negative values of $u$ or $x$ result from numerical errors and mean that one needs to reduce the time step. By such trial and error, we chose the time step $\Delta t = 0.005$. 

{ We conducted two runs of Latin Hypercube Sampling of 10,000 samples each. Four types of long-term behavior were observed:
\begin{equation*}
\begin{array}{c|c|c}
& \text{Run}~1 & \text{Run 2} \\
\hline
\text{Blow-Up} & 351 & 433 \\
\text{Equilibrial Coexistence Attractor}& 2217 & 2215\\
\text{Periodic Coexistence Attractor}& 1937 & 1902\\
\text{Predator-Free Attractor}& 5495 & 5450\end{array}
\end{equation*}
These types of behavior are illustrated and a summarising bar graph is presented in Fig. \ref{fig:LHS}.}
\begin{figure}[h!] 
    \centering    \includegraphics[width=0.9\textwidth]{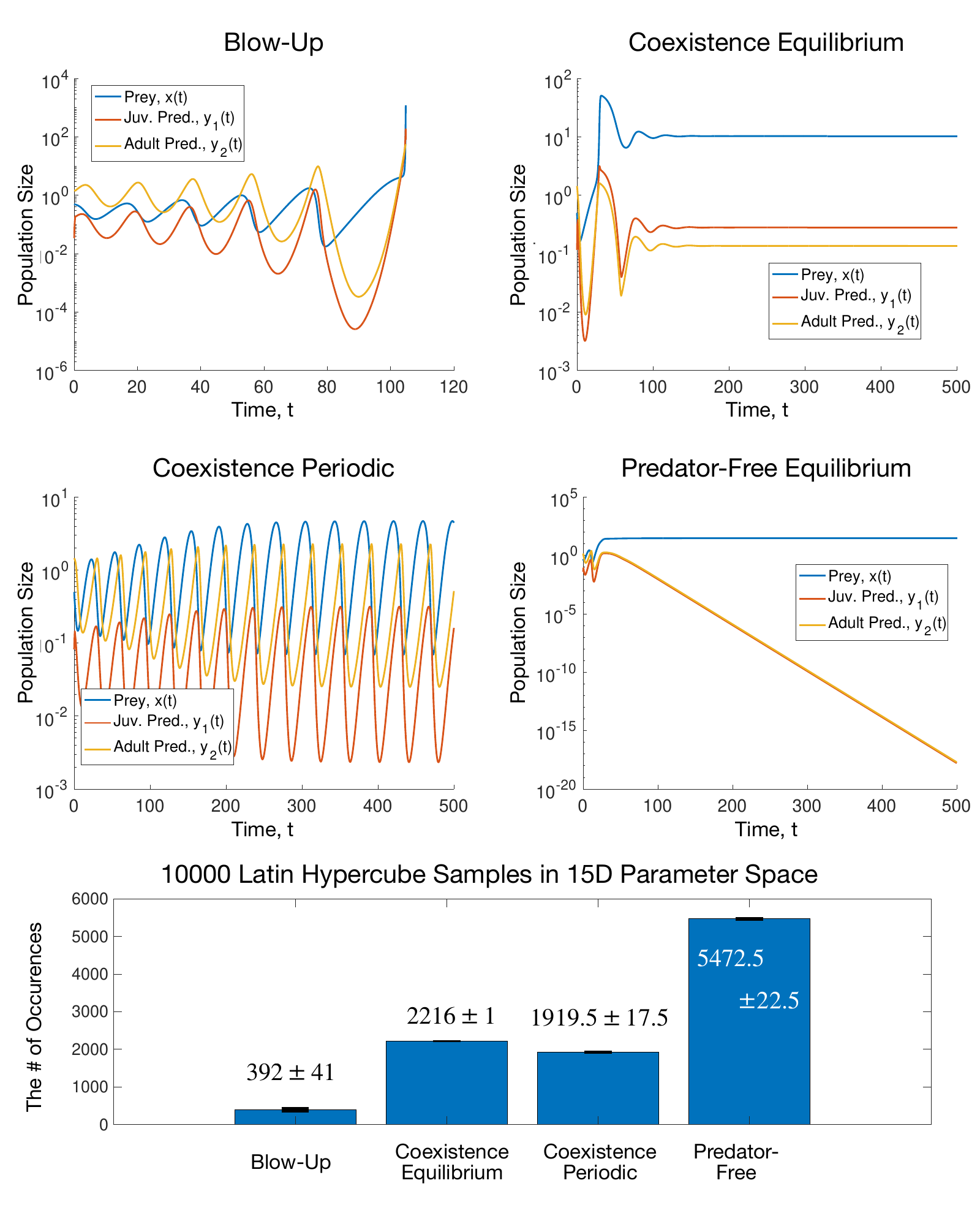}
    \caption{The results of Latin Hypercube Sampling. The intervals for the parameters are specified in Table \ref{table2}.}
    \label{fig:LHS}
\end{figure}

{ The Coexistence Equilibrium, Coexistence Periodic, and Predator-Free Attractors are biologically relevant types of behavior, while the blow-up is not. The blow-up scenarios are characterized by oscillations of the prey and predator population sizes of increasing amplitude. Such growing oscillations seem to be possible due to the unbounded predator birth-rate function $B(x,t)$ in equation \eqref{eq:birth-function}, the unbounded prey birth rate $r + sy_1$ in the prey ODE \eqref{eq:model-prey}, and time delay effects. { A positive feedback loop may arise in the model when the consumption rate of the juvenile predator, $g$, is small, the maturation age $\tau^*$ is large, and the growth rate of the prey due to eating juvenile predators, $s$, is large. Then, the more prey there are, the more prey is eaten by the adult predators, the more the juvenile predators are born, and the more prey grows due to feeding on juvenile predators without damaging them much.}
\color{black}

More biologically plausible birth rates must contain saturation, as no species can reproduce infinitely fast. For example, one may consider the following modifications to the prey ODE and the predator birth rate function,  \eqref{eq:model-prey} and \eqref{eq:birth-function}, respectively:
\begin{align}
    x' & = x\left(r-ax+s\hat{y}_1\tanh\left(\frac{y_1}{\hat{y}_1}\right) - by_2\right),\label{eq:x_saturated}\\
    B(x,\tau)& = k\hat{x}\tanh\left(\frac{x}{\hat{x}}\right)\varphi_{\tau\ge\tau^*}(\tau) + \tilde{B}(1-e^{-\zeta x}), \label{eq:birth_rate_saturated}
\end{align}
where $\hat{y}_1$ and $\hat{x}$ are parameters. The function ${\tt \tanh}$, { a popular activation function in neural network-based smooth solution models to PDEs used, e.g., in (\cite{LiLinRen})}, vanishes at zero, is approximately equal to its argument on $[0,1]$, and nearly reaches its upper bound when its argument is greater than 4. { The resulting prey birth rate term, $s\hat{y}_1\tanh\left(\sfrac{y_1}{\hat{y}_1}\right)$, is approximately equal to $sy_1$ when $y_1\le \hat{y}_1$, and is bounded from above by $s\hat{y}_1$, when $y_1\rightarrow\infty$. The predator birth rate term $k\hat{x}\tanh\left(\sfrac{x}{\hat{x}}\right)$ behaves likewise.}
We set $\hat{x} = 20$ and $\hat{y}_1 = 10$ and repeated Latin Hypercube Sampling in the 15D parameter space with parameter ranges from Table \ref{table2}. The system with saturated birth rates settled at Equilibrial Coexistence Attractor, Periodic Coexistence Attractor, and Predator-Free Attractor in 1631, 2529, and 5840 cases, respectively. No blow-up was registered. 

{The existence, uniqueness and positivity of the age-structured model with modifications \eqref{eq:x_saturated} and \eqref{eq:birth_rate_saturated} can be proven as it is done for our original age-structured model \eqref{eq:x1}--\eqref{eq:u0}. The proof of the boundedness of the solution to this modified age-structured model can be outlined as follows. The boundedness of the prey population follows from the comparison principle for ODEs and the dominance of the right-hand side of \eqref{eq:x_saturated} by the function $x(r+s\hat{y_1}) -ax^2$ for all positive $x$. Hence, the solution $x(t)$ is bounded by $a^{-1}(r + s\hat{y}_1)$. The predator birth rate function \eqref{eq:birth_rate_saturated} is bounded by $\bar{B}:=k\hat{x} + B$ from above. The predator death rate is bounded from below by $\bar{\mu}: = d_p\exp(-d_{ep}L)$ in our model. Therefore, the predator population will be dominated by the solution to the KMRE and hence will be bounded. We leave further improvements of the age-structured model and a rigorous proof of its properties for future work.}

%Checking analytically that the solutions to the initial-value problem for the age-structured model with saturated birth rates \eqref{eq:x_saturated}--\eqref{eq:birth_rate_saturated} are bounded is difficult due to the complex interplay between the prey and predator and time delay effects. We leave it to future work.

%  Saturated model: eq = 1631, per = 2529, pred-free = 5840

\subsection{Linear Discriminant Analysis}
Linear Discriminant Analysis (LDA) (or Multiple Discriminant Analysis (MDA)) is a classical linear supervised learning tool (\cite{DHS2001}). It aims at finding a low-dimensional subspace such that data from different categories projected onto this subspace are separated the most while the projected data from the same categories are clustered the most. A description of LDA is found in Appendix \ref{sec:AppB}. 

We use LDA to discriminate parameter sets with four different types of long-term behavior and find which parameters affect the type of long-term behavior the most. Since there are four categories, the maximal dimension of the optimal subspace is three. We project the parameters onto an optimal plane spanned by the two dominant eigenvectors of the generalized eigenvalue problem described in  Appendix \ref{sec:AppB} -- see Fig. \ref{fig:LDA}.
\begin{figure}[h!] 
    \centering    \includegraphics[width=\textwidth]{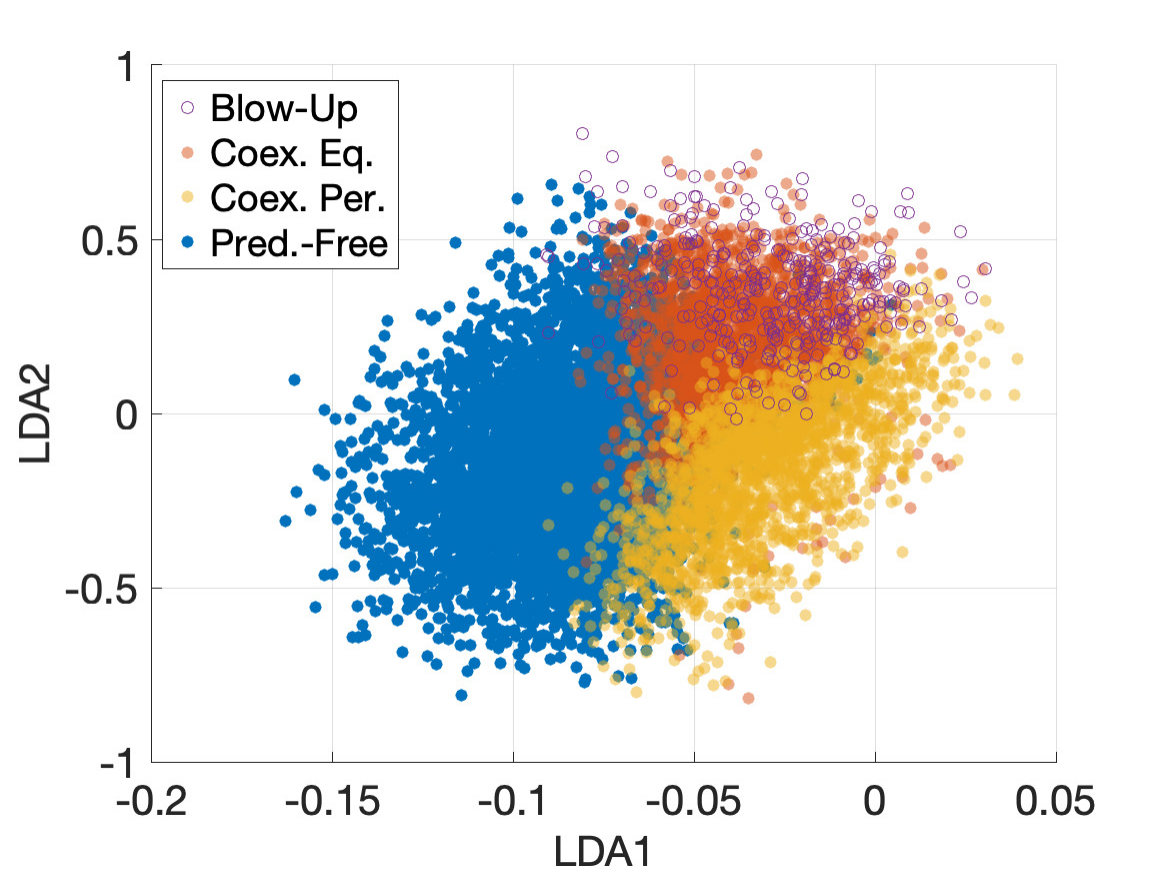}
    \caption{LDA projection from the 15D parameter space onto a 2D optimal subspace where the parameter sets leading to four different types of long-term behavior (Blow-Up, Coexistence Equilibrium, Coexistence periodic, and Predator-free attractor) are separated the most.}
    \label{fig:LDA}
\end{figure}
Note that the eigenvectors of the LDA generalized eigenvalue problem are not orthogonal. We orthonormalized them by one step of the Gram-Schmidt procedure.

To understand which parameters affect the type of long-term behavior the most, we do the following calculation. Let $X$ be $n\times d$ matrix whose rows are the parameter sets generated by the Latin Hypercube Sampling, $n = 10,000$, and $d = 15$ in our case. Let $R$ be a diagonal matrix whose diagonal entries $R_i$ are the differences between the maximal and minimal values of parameter $i$, $i=1,\ldots,d$. Then, the matrix $X$ can be decomposed as 
\begin{equation}
    \label{eq:Xdecomp}
    X = ZR + 1_{n\times 1}x_{\min},
\end{equation}
where $x_{\min}$ is a row vector whose entries are the lower bounds for the parameter ranges, and $Z$ is an $n\times d$ matrix whose entries take values between zero and 1. Therefore, to account for ranges, we take the $d\times 2$ LDA projection matrix $W$ with orthonormalized columns and multiply it by $R$ on the left. Then, we plot bar graphs of $RW$ and display the result in Fig. \ref{fig:LDA_vec}. 
\begin{figure}[h!] 
    \centering    
    \includegraphics[width=\textwidth]{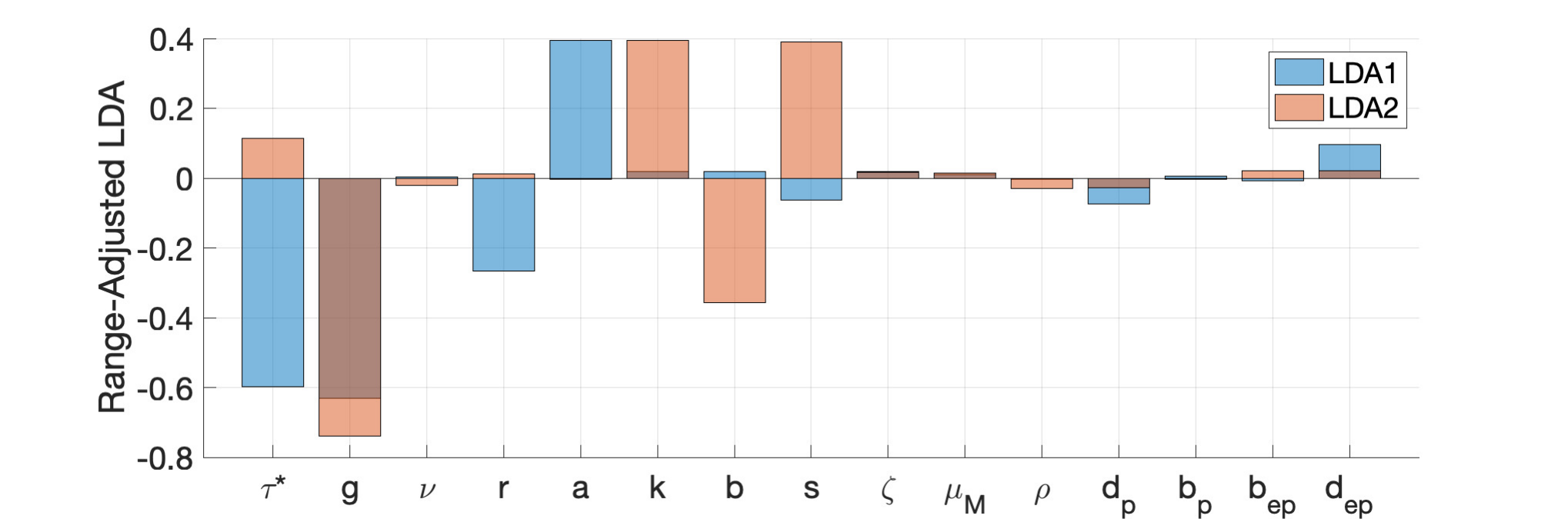}
    \caption{LDA range-adjusted bar graphs quantifying the importance of the model parameters for the type of long-term behavior.}
    \label{fig:LDA_vec}
\end{figure}
This graph suggests that the first LDA is the most influenced, in decreasing order, by  $g$, the consumption rate of juvenile predators by the prey, $\tau^*$, predator maturation age, $a$, intraspecific competition rate of the prey, and $r$, the prey birth rate. The second LDA is the most affected, in decreasing order, by $g$, $k$, the reproduction rate of the predator, $s$, the consumption rate of prey by adult predators, and $b$, the rate of predation by the predator.

\subsection{Phase diagrams in the $(\tau^*,g)$-plane} 
Fig. \ref{fig:LDA_vec} suggests that the two most influential parameters of our model on the long-term behavior type are the predator maturation age, $\tau^*$, and the consumption rate of juvenile predators by the prey, $g$. We conducted a detailed investigation of the long-term behavior of the system \eqref{eq:model-prey}--\eqref{eq:smooth_phi} at $(\tau^*,g)\in[0,{2}]\times[0,1]$, and the rest of the parameters fixed at the selected values specified in the third column of Table \ref{table2}. { We used a time/age step of $h = 0.0125$}.
%, except for the indicator function smoothness parameter $\nu$ that assumed three values: $\nu\in\{1,10,100\}$. 
%These parameter settings admit a large time step on $\Delta t = 0.1$ in numerical integration. 

For each value of $\tau^*$ from 0.1 to {2} with step { 0.05}, we moved along the parameter $g$ with step 0.01 from $g=1$ down to $g=0$. At each value of $g$, we found the equilibrium as described in Section \ref{sec:find_equilib}. If the equilibrium was unstable, i.e., if the Jacobian \eqref{eq:J1}--\eqref{eq:Jk} had an eigenvalue with a positive real part, we found the Periodic Coexistence attractor as described in Section \ref{sec:find_limit_cycles}. { We did so for two values of the indicator function smoothness parameter $\nu$: $\nu = 100$ corresponding to a sharp transition from juvenile to adult, and $\nu = 1$, describing a gradual transition. 
The resulting phase diagrams in the plane $(\tau^*,g)$, { bifurcation diagrams at $\tau^* = 1$,}  and ensembles of bifurcation diagrams in $g$ corresponding to the grid values of $\tau^*$ are displayed in Fig. 
\ref{fig:nu100}. Comparing these diagrams at $\nu = 100$ and $\nu = 1$, we observe that the gradual transition from juvenile to adult somewhat increases the region of the Predator-Free Attractor, slightly reshapes the region of the Periodic Coexistence Attractor, and slightly increases the amplitude of stable oscillations at large $\tau^*$ compared with those in the case of the sharp transition. }  
{ We also observe an oscillatory feature in the upper branches of the bifurcation diagrams for the juvenile predator at both $\nu = 100$ and $\nu = 1$. These features persisted throughout our time- and age-step refinement. We leave an investigation into this phenomenon for future work.}

\begin{figure}[h!] 
    \centering    
    \includegraphics[width=\textwidth]{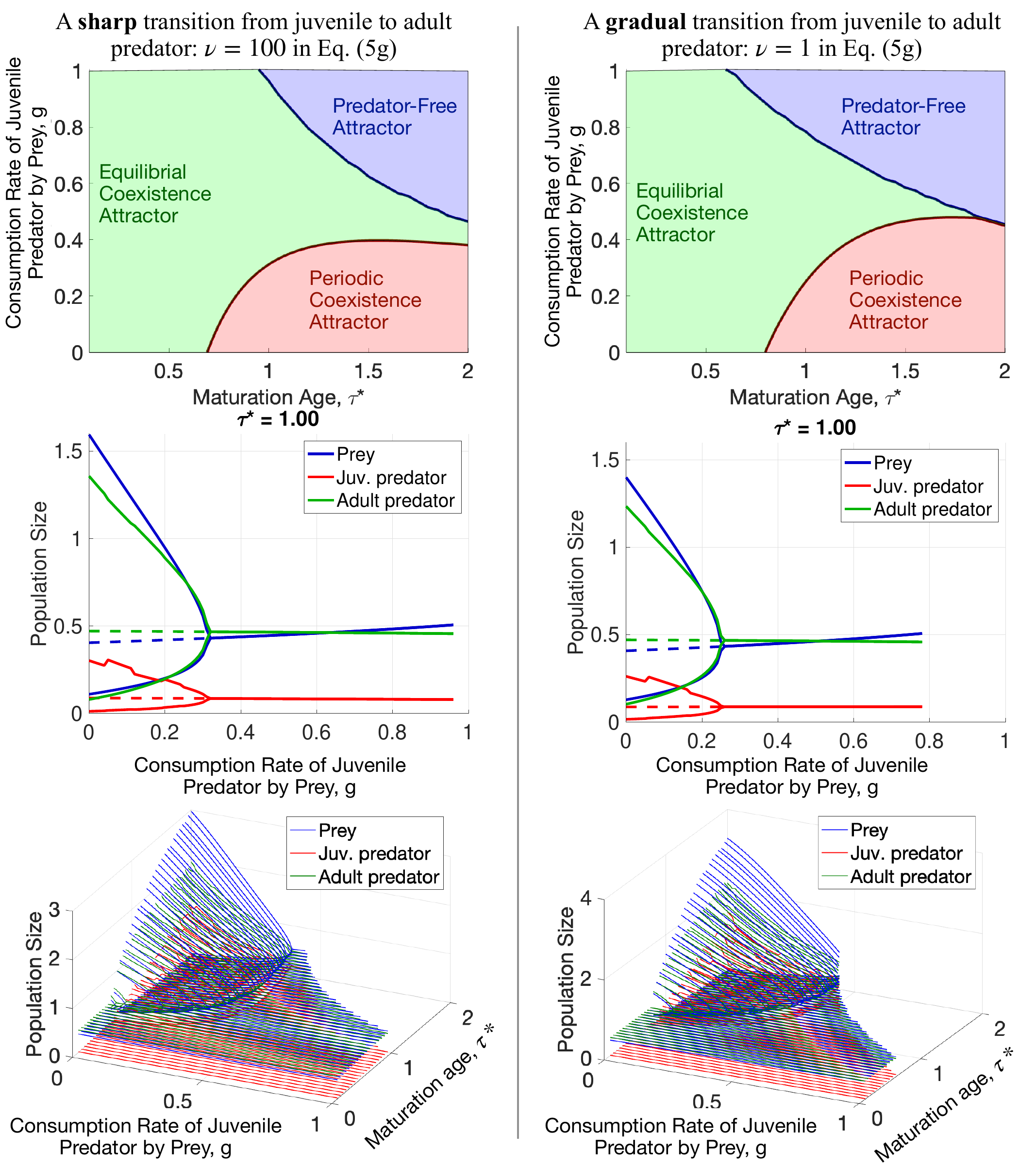}
    \caption{  Phase and bifurcation diagrams for the proposed model \eqref{eq:model-prey}--\eqref{eq:smooth_phi}. A sharp (left) and a gradual (right) transition from juvenile to adult predator: $\nu = 100$  and $\nu = 1$ in \eqref{eq:smooth_phi}, respectively. Top row: Phase diagrams. Middle Row: Bifurcation diagrams at the maturation age $\tau^* =1$. Bottom row: Ensembles of the bifurcation diagrams.}
    \label{fig:nu100}
\end{figure}

\subsection{Age density}
We examined predator age-density at several pairs of maturation age $\tau^*$ and the consumption rate of juvenile predators by prey $g$ at which the system admits the Equilibrial Coexistence Attractor. The remaining parameters were set to their selected values in Table \ref{table2}. The results are shown in Fig. \ref{fig:AgeDenEq}. The predator age-density decays rapidly with age and approaches zero at the age cutoff $L = 30$, as desired.
\begin{figure}[h!] 
    \centering    
    \includegraphics[width=\textwidth]{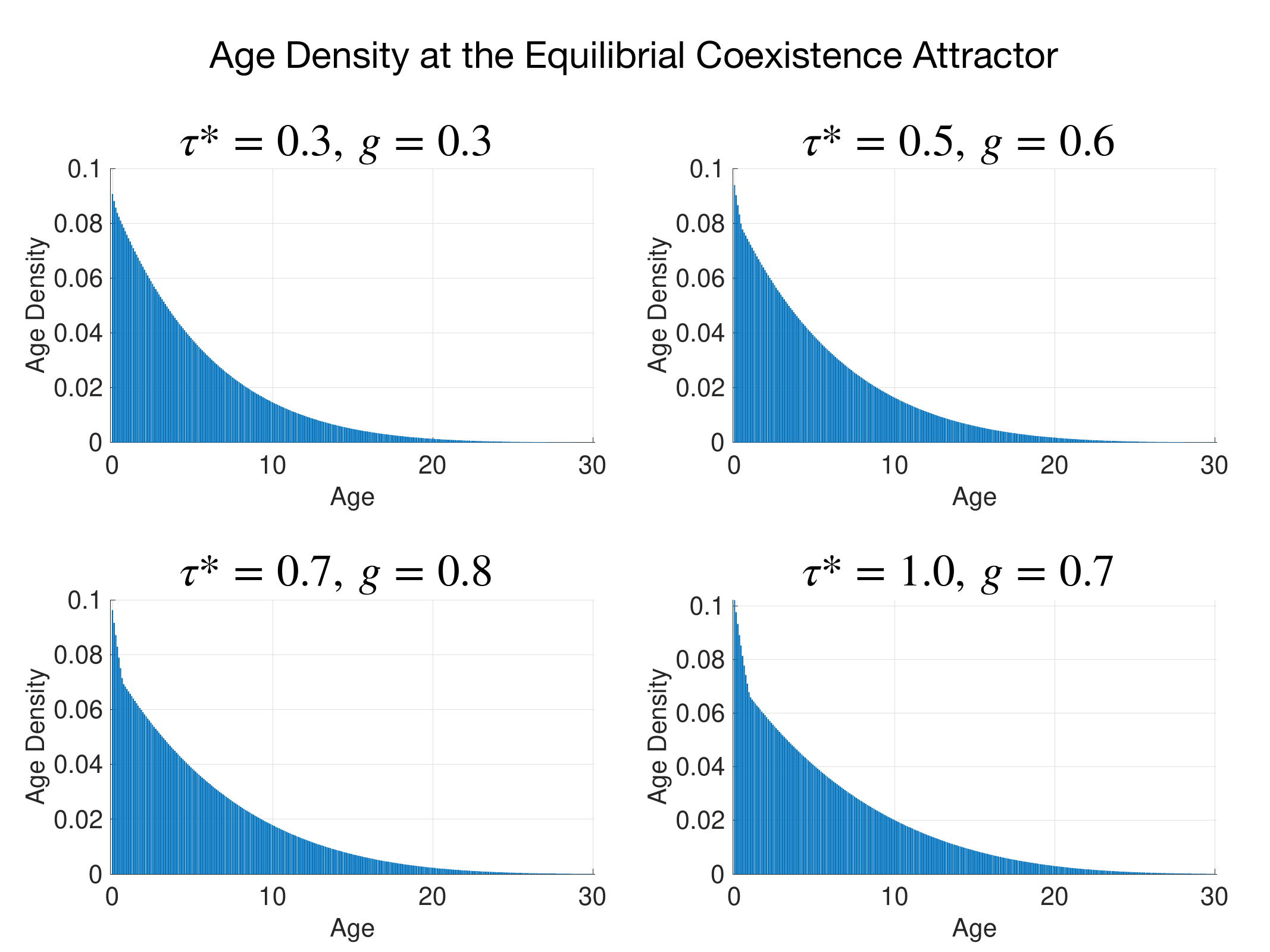}
    \caption{The predator age density at several pairs of values $(\tau^*,g)$ leading to the Equilibrial Coexistence Attractor. The rest of the parameters are at their selected values from Table \ref{table2}, and $\nu = 100$.}
    \label{fig:AgeDenEq}
\end{figure}
We have also extracted the predator age density at $(\tau^* = 1, g = 0.1)$ where the system settles on the Periodic Coexistence Attractor -- see Fig. \ref{fig:AgeDenPer}.
\begin{figure}[h!] 
    \centering    
    \includegraphics[width=\textwidth]{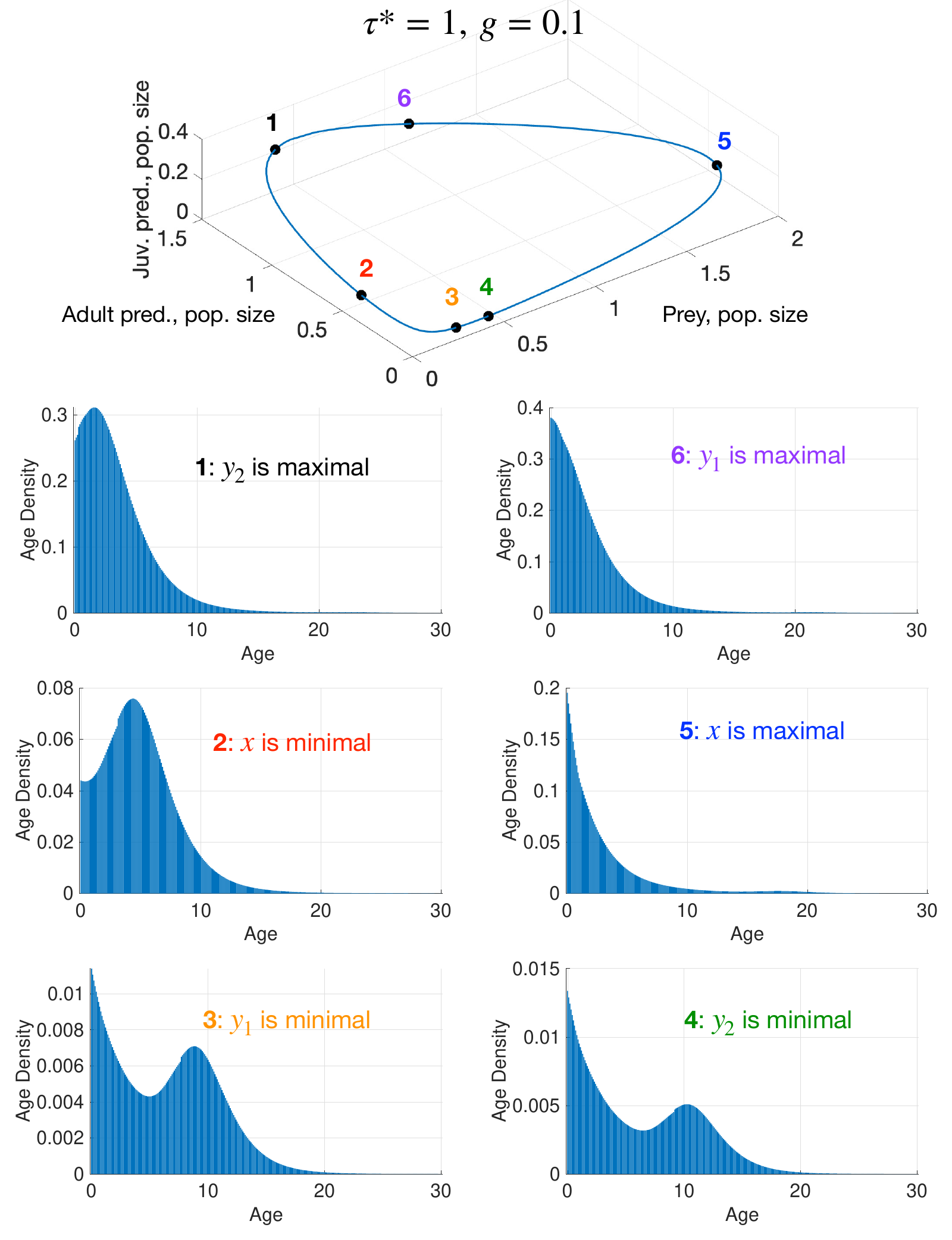}
    \caption{The predator age density at $(\tau^* = 1, g = 0.1)$ leading to the Periodic Coexistence Attractor. The rest of the parameters are at their selected values from Table \ref{table2}, and $\nu = 100$.}
    \label{fig:AgeDenPer}
\end{figure}
We observe rapidly decaying age density waves and the following alternation of peaks and minima of the prey and juvenile and adult predator population sizes: 
\begin{equation*}
\max y_2\rightarrow
\min x\rightarrow
\min y_1\rightarrow
\min y_2\rightarrow
\max x\rightarrow
\max y_1\rightarrow \max y_2.
\end{equation*}

\subsection{Comparison with the ODE and DDE models}
\label{sec:ODEnum}
The proposed model \eqref{eq:model-prey}--\eqref{eq:smooth_phi} was reduced to ODE and DDE models,  \eqref{eq:Luis_prey}--\eqref{eq:Luis_y2}  and \eqref{eq:xDDE0}--\eqref{eq:uttau*1}, respectively, by replacing the smooth indicator functions with the sharp ones, integration over age, and age-averaging (see Section \ref{sec:conn2ODE}). The trapezoidal rule was used to approximate the transition term from juvenile to adult predator in the DDE model, yielding delayed terms by the maturation age $\tau^*$. 

The goal of this section is to numerically investigate the ODE  and DDE models and compare their long-term behavior to that of the age-structured model \eqref{eq:model-prey}--\eqref{eq:smooth_phi}.  
We set all parameters of \eqref{eq:model-prey}--\eqref{eq:smooth_phi}, except for $\tau^*$ and $g$, to their selected values from Table \ref{table2}. The parameters $(\tau^*,g)$ ran through all values from the rectangle $[0,2]\times[0,1]$. The  parameters for the ODE system \eqref{eq:Luis_prey}--\eqref{eq:Luis_y2}, 
\begin{itemize}
\item $D$, the rate from juvenile to adult predator,~\eqref{eq:D_ode},
\item $b_2$, the predator birth rate,~\eqref{eq:b2_ode},
\item $m_1$, the death rate for juvenile predators,~\eqref{eq:m1_ode}, and
\item $m_2$, the death rate for adult predators,~\eqref{eq:m2_ode},
\end{itemize}
were computed by { age-averaging at the equilibria, stable or unstable,} admitted by the system \eqref{eq:model-prey}--\eqref{eq:smooth_phi} { with $\nu = 100$} at each pair $(\tau^*,g)$. 
The age-averaged values of $D$, $b_2$, $m_1$, and $m_2$ as functions of $\tau^*$ and $g$ are displayed in Fig. \ref{fig:plots-averaged-parameters}. { Their ranges are
\begin{equation*}
    \begin{array}{ll}
    D:&\quad [0.353, 9.91],\\
b_2:&\quad [0.0794, 0.0835],\\
m_1:&\quad [0.0200, 0.0219],\\
m_2:&\quad [0.0363, 0.0546].
    \end{array}
\end{equation*}}
The DDE model \eqref{eq:xDDE0}--\eqref{eq:uttau*1} uses the same parameter values $b_2$, $m_1$, and $m_2$ as the ODE model. It does not involve the parameter $D$. 

Thus, with the selected parameter values from Table \ref{table2} and Fig. \ref{fig:plots-averaged-parameters}, the ODE and DDE systems mimic the age-structured system \eqref{eq:model-prey}--\eqref{eq:smooth_phi} as closely as possible. 
\begin{figure}[h!]
    \centering
\includegraphics[width = \textwidth]{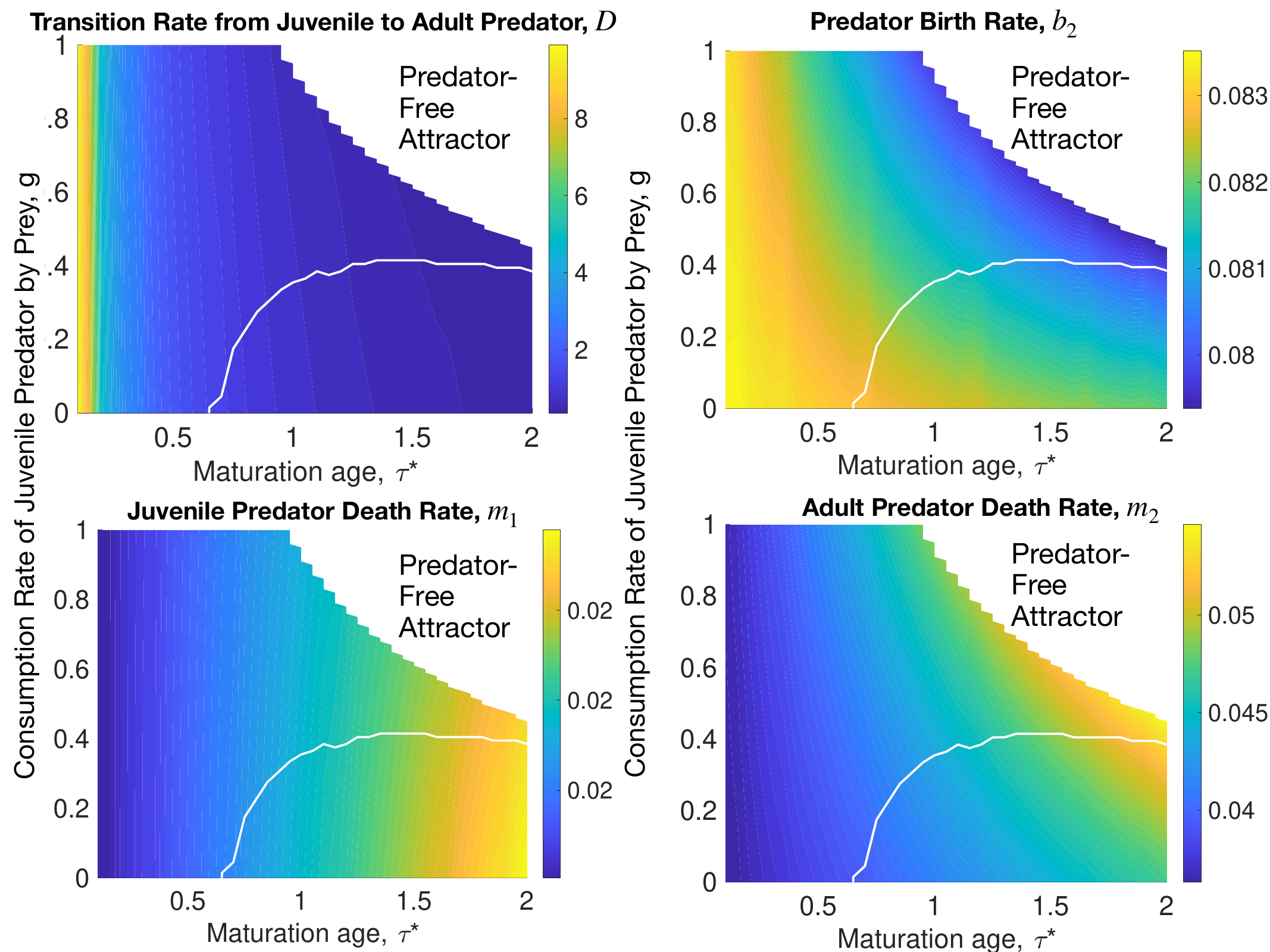}
\caption{The dependence of the age-averaged parameters of the ODE model \eqref{eq:Luis_prey}--\eqref{eq:Luis_y2} on the maturation age $\tau^*$, and the consumption rate of juvenile predators by the prey, $g$.  
}
\label{fig:plots-averaged-parameters}
\end{figure}

The phase diagrams in $(\tau^*,g)$ for the ODE system \eqref{eq:Luis_prey}--\eqref{eq:Luis_y2} and the DDE system \eqref{eq:xDDE0}--\eqref{eq:uttau*1}  are shown in Fig. \ref{fig:diagram_ODE_DDE}. 

As in the age-structured model, the phase diagrams for the ODE and DDE systems have three regions corresponding to the Equilibrial Coexistence Attractor, the Periodic Coexistence Attractor, and the Predator-Free Attractor. The regions of the Predator-Free Attractor for the ODE and DDE models did not increase compared to the age-structured model \eqref{eq:model-prey}--\eqref{eq:smooth_phi} in Fig. \ref{fig:nu100} (top left). We remark that we did not run the ODE and DDE models in the Predator-Free region of the age-structured model with $\nu = 100$ because we did not compute the parameters $D$, $b_2$, $m_1$, and $m_2$ in it.}  The region of the Periodic Coexistence Attractor is considerably smaller for the ODE system. This region for the DDE system is intermediate in size between those of the ODE and age-structured models. Fig. \ref{fig:phase_diagram_comparison} superimposes the phase diagrams of the age-structured model \eqref{eq:model-prey}--\eqref{eq:smooth_phi} with the indicator function smoothness parameters $\nu = 1$ and $100$, the ODE system  \eqref{eq:Luis_prey}--\eqref{eq:Luis_y2}, and the DDE system \eqref{eq:xDDE0}--\eqref{eq:uttau*1}.  The age-structured model is much more prone to developing oscillations than the closely mimicking ODE and DDE systems. 

\begin{figure}[h!] 
    \centerline{
(a)    \includegraphics[width=\textwidth]{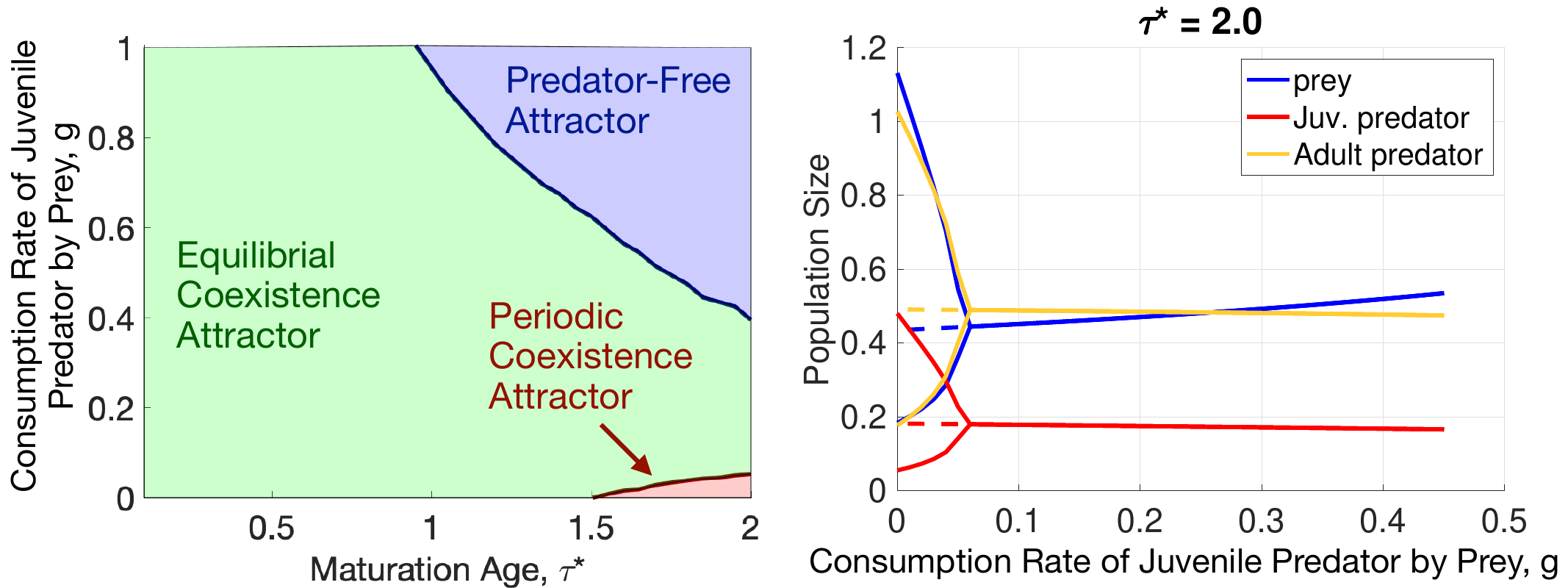}
}
    \centerline{
(b)    \includegraphics[width=\textwidth]{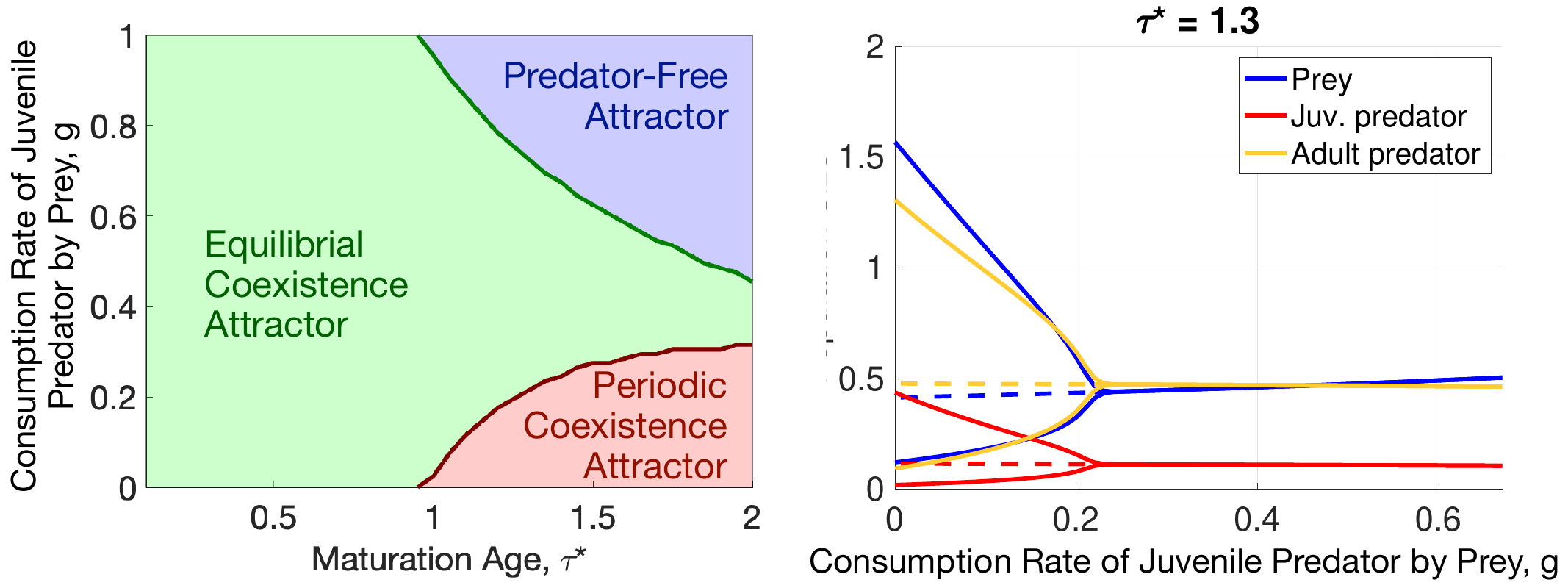}
}
    \caption{(a): The phase diagram in the $(\tau^*,g)$ plane and a bifurcation diagram at $\tau^* = 2$ for the ODE model \eqref{eq:Luis_prey}--\eqref{eq:Luis_y2} with parameters $D$, $b_2$, $m_1$, and $m_2$ age-averaged over the corresponding attractors of the proposed model \eqref{eq:model-prey}--\eqref{eq:smooth_phi} and displayed in Fig. \ref{fig:plots-averaged-parameters}.
    (b): The phase diagram in the $(\tau^*,g)$ plane and a bifurcation diagram at $\tau^* = 1.3$ for the DDE model \eqref{eq:xDDE0}--\eqref{eq:uttau*1} where parameters $b_2$, $m_1$, and $m_2$ age-averaged over the corresponding attractors of the proposed model \eqref{eq:model-prey}--\eqref{eq:smooth_phi} and displayed in Fig. \ref{fig:plots-averaged-parameters}.
    }
    \label{fig:diagram_ODE_DDE}
\end{figure}

\begin{figure}[h!] 
    \centering    \includegraphics[width=\textwidth,height=\textheight,keepaspectratio]{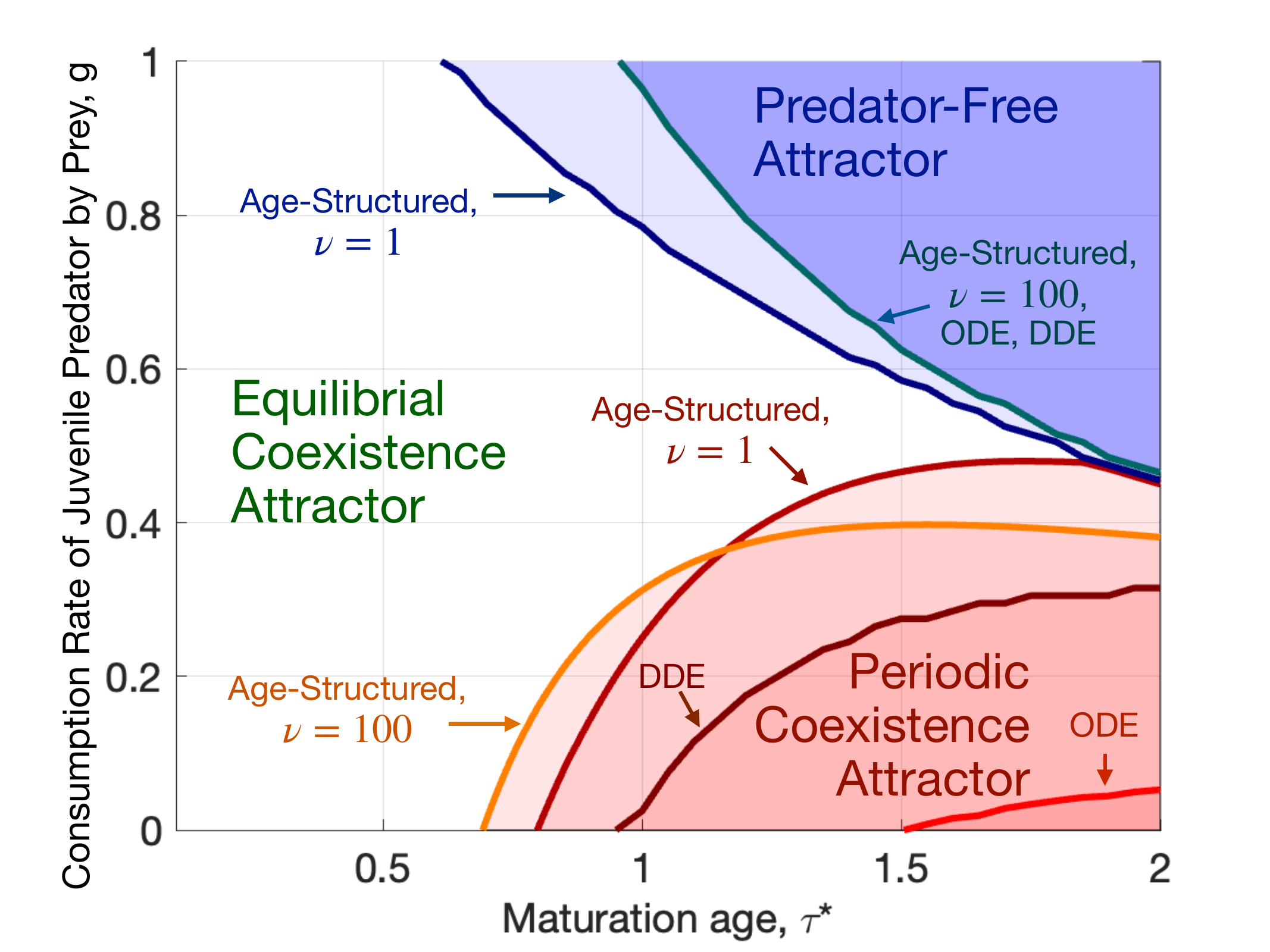}
    \caption{A comparison of phase diagrams on the $(\tau^*,g)$ plane of the age-structured model \eqref{eq:model-prey}--\eqref{eq:smooth_phi} with the indicator function smoothness coefficients $\nu = 100$ and $1$, the ODE model \eqref{eq:Luis_prey}--\eqref{eq:Luis_y2}, and the DDE model \eqref{eq:xDDE0}--\eqref{eq:uttau*1} with parameters mimicking the age-structured model with $\nu=100$.}
    \label{fig:phase_diagram_comparison}
\end{figure}

\section{Discussion}
\label{sec:discussion}
We developed a predator-prey model \eqref{eq:model-prey}--\eqref{eq:smooth_phi} with an age-structured predator population and role reversal, proved the existence, uniqueness, and positivity of the initial-value problem for it, and investigated its long-term behavior numerically. In addition, we derived ODE and DDE  models from it and compared their long-term behavior to that of the age-structured model with the corresponding parameter values. This study taught us several important lessons.
\begin{enumerate}
    \item The proposed model involves 15 parameters. Latin Hypercube Sampling has shown that, depending on these parameters, the system settles on the Predator-Free Attractor, Equilibrial Coexistence Attractor, and Periodic Coexistence Attractor in 22\%, 19\%, and 55\% of cases, respectively, or blows up in approximately 4\% of cases. The blow-up is not biologically relevant. {It is enabled due to the terms $kx\varphi_{\{\tau\ge \tau^*\}}(\tau)$ in the birth rate function  of the predator, \eqref{eq:birth-function} and $sxy_1$ in the right-hand side of the prey ODE \eqref{eq:model-prey}. The replacement of the birth-rate function and the prey ODE with their saturated versions \eqref{eq:x_saturated} and \eqref{eq:birth_rate_saturated} resulted in the absence of blow-ups in the Latin Hypercube Sampling. Nonetheless, we left this caveat in our model for educational purposes. Our future models will necessarily involve saturation, preventing blow-ups. }
    
    \item We performed a Linear Discriminant Analysis to find which parameters affect the type of long-term behavior the most. These are the maturation age $\tau^*$ and the consumption rate of the juvenile predator by the prey $g$. {The importance of these parameters agrees with the main claim of \cite{WG1984}: the maturation age of the predator and its ontogenetic niche shift in diet and interspecific interactions are significant determinants of the dynamics of simplified ecological systems.}
    \item We fixed all parameters but $\tau^*$ and $g$ at their selected values (see Table \ref{table2}) and further investigated the system's dynamics. Blow-up does not occur in these settings, and the system approaches one of the three attractors depending on $\tau^*$ and $g$. 
    { The phase diagrams in Fig. ~\ref{fig:nu100} in the $(\tau^*,g)$-plane reveals the following. If the maturation age is small enough, the system reaches a coexistence equilibrium at all values of $g\in[0,1]$. As the maturation age $\tau^*$ increases, the predator becomes extinct at large $g$, and stable oscillations develop at small $g$. The region of the Predator-Free Attractor occupies the top right corner of the phase diagrams where both $\tau^*$ and $g$ are large, while the Periodic Coexistence Attractor occupies the bottom right corner where $\tau^*$ is large and $g$ is small. This suggests that the long maturation of the predator and significant consumption of its juveniles create a juvenile bottleneck that is hard to overcome and results in the predator's extinction. This result is in agreement with~\cite{WG1984}.  }
    
\item { The effect of a gradual rather than a sharp transition from juvenile to adult is a minor increase of the Predator-Free Attractor region, a minor reshaping of the Periodic Coexistence Attractor region,   and a minor increase of the amplitude stable oscillations at large $\tau^*$. Overall, this effect is mild. }

    \item We extracted the predator age density at a collection of pairs $(\tau^*,g)$ leading to an Equilibrial Coexistence Attractor and at a pair $(\tau^*,g)$ admitting the Periodic Coexistence Attractor. The age densities at Equilibrial Coexistence Attractors monotonically decay, while they feature decaying waves at the periodic Coexistence Attractor. This behavior of the age density is biologically relevant.
    \item To gauge the effect of the age-structured predator population, we derived ODE and DDE models. Furthermore, we computed the parameter values for these models by age-averaging the solution to the age-structured model at its coexistence equilibrium, whether stable or unstable. The phase diagram for the resulting ODE model has a much smaller area of the Periodic Coexistence Attractor and a much larger region of the Equilibrial Coexistence Attractor. The sizes of these regions in the phase diagram of the DDE model are intermediate between those of the ODE and age-structured models. Therefore, the age structure of the predator population promotes oscillatory behavior. 
    \item { Numerical simulation of the ODE and DDE models is much faster than that of the age-structured model. Time and age step refinement in the age-structured model by a factor of two increases runtime by a factor of four. The computation time for the phase diagrams in Fig. 10, which involve finding equilibria and limit cycles, was several days for the age-structured model, compared with several hours for the ODE and DDE models. }
\end{enumerate}
{ The finding that different models of the same system can yield qualitatively different predictions is not new. For example, \cite{Cantrell2012} showed that spatially explicit and spatially implicit models of a one-dimensional two-patch system yield different predictions for the survival of a population migrating between the two patches.}

\section{Conclusion}
\label{sec:conclusion}
{
{
The proposed predator-prey model \eqref{eq:model-prey}--\eqref{eq:smooth_phi} with a role reversal and age-structured predator population captures the role of predators' gradual maturation and the ontogenetic niche shifts in their diet and their interaction with the prey. The long-term behavior of this system at the default parameter values agrees with our knowledge about some real predator-prey systems. A comparison of the long-term behavior of the age-structured model with the mimicking ODE- and DDE-based models demonstrates that the model type strongly affects the type of coexistence attractor the system admits. Because the age-structured model includes elements of biological reality that are absent in the ODEs and DDEs, one should use caution when drawing conclusions from the latter model types. %Interestingly, the predictions for the extinction of predators are not affected by the model type. 

The proposed model contains a number of imperfections. First, it admits blow-up for a small fraction of parameter cases, and we have suggested some strategies for revising the model that should avoid this. Second, the model does not account for Allee effects~(\cite{Allee1949}), demographic stochasticity, and other ecological processes important at small population sizes. Consequently, the lower bounds of oscillatory solutions to the coupled predator-prey system may be very low, which is ecologically implausible. Solving this issue would require major changes to the model formulation. Nevertheless, we have learned a lot from our study of this model and leave further improvements to future work.
}
\color{black}

\section*{Acknowledgements}
The work of M.C. was partially supported by the AFOSR MURI grant FA9550-20-1-0397. W.F.F. acknowledges support from the U.S. National Science Foundation (DMS2451241). The work of D.L. was partially supported by the Simons Foundation.

\section*{Declarations}

\begin{itemize}
\item Funding: {AFOSR MURI grant FA9550-20-1-0397 (MC);  Simons Foundation Award 848629 (DL); U.S. National Science Foundation Award DMS2451241 (WFF).}
\item Conflict of interest/Competing interests: { The authors have no conflict of interest.}
\item Ethics approval and consent to participate: { The authors comply with ethics rules.}
\item Consent for publication: {All authors consent for publication.}
\item Data availability: {Due to the volume of the simulation data, it will be made available upon request. }
\item Materials availability: {Not applicable.}
\item Code availability: { Codes are available on GitHub~(\cite{mar1akc}).}
\item Author contribution: {Luis Suarez: Conceptualization, Methodology, Software, Validation, Formal analysis, Investigation, Data curation, Writing – original draft, Visualization. Maria Cameron: Methodology, Software, Validation, Formal analysis, Investigation, Data curation, Writing – review\&editing, Visualization, Supervision, Funding Acquisition. William Fagan: Conceptualization, Writing – review\&editing, Supervision. Doron Levy: Conceptualization, Writing – review\& editing, Supervision, Project administration, Funding Acquisition.}
\end{itemize}

%\section{Section title of first appendix}\label{secA1}
%

 \appendix
\setcounter{equation}{0}
\renewcommand{\theequation}{\Alph{section}-\arabic{equation}}
    \setcounter{lemma}{0}
    \renewcommand{\thelemma}{\Alph{section}\arabic{lemma}}

\section*{Appendix}

\section{ Proof of Theorem \ref{thm1}} 
\label{sec:proof}
First of all, we remind the reader that the age-structured model \eqref{eq:model-prey}--\eqref{eq:smooth_phi} can be rewritten as a system of infinitely many ODEs, \eqref{eq:x1}--\eqref{eq:u0}. The first ODE, \eqref{eq:x1}, describes the dynamics of the prey population size, while the two infinite families of ODEs, \eqref{eq:u1} and \eqref{eq:u2}, describe the dynamics of generations of predators present at the initial time $t=0$ and born at $t>0$,  respectively.

\subsection{Compact support of the predator age density.}
\begin{lemma}
    The age density of the predator is compactly supported at any finite time $t$.
\end{lemma}
\begin{proof}
The initial age density $u_0(\tau)$ is bounded and has a compact support. The age density $u(t,\tau)$ propagates along characteristics \eqref{eq:u1}--\eqref{eq:u2}. Equation \eqref{eq:u1} implies that $u(t,t+\tau_0) = 0$ whenever $u(0,\tau_0) = u_0(\tau_0) = 0$. Hence, if $u_0(\tau) = 0$ on $[\tau_{0,\max},+\infty)$ then at a fixed $t$, $u(t,\tau) = 0$ on $[\tau_{0,\max}+t,+\infty)$. 
\end{proof}
For brevity, we denote the maximal possible age at time $t$ by $\tau_{\max}(t)$:
\begin{equation}
    \label{eq:taumax}
    \tau_{\max}(t): = \tau_{0,\max} +t. 
\end{equation}

\subsection{A Banach space for prey population size and predator age density.} 
To analyze the solutions to the initial value problem \eqref{eq:model-prey}--\eqref{eq:smooth_phi}, we define a space for the prey population size $x$ and the predator age density $u$.

\begin{lemma}
The space
 \begin{equation}
    \label{eq:Xdef}
    \mathcal{X} = \{(x,u)~|~x\in\mathbb{R},~u\in L_1([0,\tau_{\max}(t)])\}
\end{equation}
with the norm 
\begin{equation}
    \label{eq:normX}
    \|(x,u)\|_{\mathcal{X}} = |x| + \int_{0}^{\tau_{\max}(t)} |u|d\tau.
\end{equation}
 is a Banach space at each fixed finite time $t$.  
\end{lemma}
\begin{proof}
    Indeed, $L_1([0,\tau_{\max}(t)])$ is Banach and hence $\mathcal{X}$ is a direct product of Banach spaces -- see~(\cite{AdamsSobolev2003}).
\end{proof}

\subsection{Lipschitz continuity.}
\begin{lemma}
\label{lemma:Lip}
    The right-hand side $F:=(f_1,f_2)$ of \eqref{eq:x1}--\eqref{eq:u2} is locally Lipschitz with respect to the norm \eqref{eq:normX}, i. e. for any
 \begin{equation}
     (x,u),~(y,v)\quad\text{such that}\quad\|(x,u)\|_{\mathcal{X}} \le R,~\|(y,v)\|_{\mathcal{X}} \le R,
 \end{equation}
 there exists a Lipschitz constant $L_{F,R}$ dependent on $F$ and $R$ but independent of $(x,u)$ and $(y,v)$ such that
  \begin{equation}
  \label{eq:Flip}
     \|F(x,u)-F(y,v)\|_{\mathcal{X}} \le L_{F,R}\|(x,u)-(y,v)\|_{\mathcal{X}}.
 \end{equation}
\end{lemma}
    
\begin{proof}
We will compute 
 \begin{equation}
     |f_1(x,u)-f_1(y,v)| \quad{\rm and}\quad \|f_2(x,u)-f_2(y,v)\|_{L_1([0,\tau_{\max}(t)]}.
 \end{equation}
 For brevity, we will use the notation 
 \begin{equation}
   \|f_2(x,u)-f_2(y,v)\|_1: =  \|f_2(x,u)-f_2(y,v)\|_{L_1([0,\tau_{\max}(t)]}.
 \end{equation}
 We start with $f_1$, the right-hand side of \eqref{eq:x1}:
\begin{align}
    &\thinspace |f_1(x,u)-f_1(y,v)|  =   \notag \\  
    =\thinspace & \thinspace \left|r(x-y) - a(x^2-y^2)  +s\left(x\int_0^{\tau^*} u(t,\tau)d\tau - y \int_0^{\tau^*} v(t,\tau)d\tau\right)\right. \notag \\
    &\quad-\thinspace  \left.b\left(x\int_{\tau^*}^{\tau_{\max}(T)} u(t,\tau) d\tau - y\int_{\tau^*}^{\tau_{\max}(T)} v(t,\tau) d\tau\right)\right| \notag \\ \leq\thinspace & |x-y|\left|r + a(x+y)\right| +s\left|x\int_0^{\tau^*} u(t,\tau)d\tau - y \int_0^{\tau^*} v(t,\tau)d\tau\right| \notag \\
    &\quad+\thinspace  b\left|x\int_{\tau^*}^{\tau_{\max}(T)} u(t,\tau) d\tau - y\int_{\tau^*}^{\tau_{\max}(T)} v(t,\tau) d\tau\right| 
    \notag \\ 
    \le \thinspace & |x-y|\left|r + a(x+y)\right| + \max\{s,b\}|x|\int_0^{\tau_{\max(t)}}|u(t,\tau)-v(t,\tau)| d\tau \notag \\
    &\quad+\thinspace   \max\{s,b\}|y-x|\int_0^{\tau_{\max}(t)}|v(t,\tau)|d\tau \notag \\
   \le \thinspace & |x-y|\left(r+2aR\right) +\max\{s,b\} R\left(|y-x| + \|u-v\|_1\right).
   \label{eq:lip-for-f1}
\end{align}
We continue with $f_2(x,u) = -\mu(x,\tau) u$, the right-hand side  of \eqref{eq:u1} and \eqref{eq:u2}, where the death rate function $\mu(x,\tau)$ is given by \eqref{eq:mufun}. The base death rate term in \eqref{eq:mufun}, $\mu_B(\tau)$, is continuous and bounded at every finite $\tau$ by the statement of Theorem \ref{thm1}. Therefore, we assume that $0\le \mu_B(\tau)\le M_B$ on $0\le \tau\le\tau_{\max(t)}$. Hence,
\begin{align}
& \|f_2(x,u)-f_2(y,v)\|_1 =\int_{0}^{\tau_{\max}(T)}|f_2(x,u)-f_2(y,v)|d\tau \notag \\
\le\thinspace & g |x|\int_{0}^{\tau^* }|u(t,\tau)-v(t,\tau)|d\tau + g|x-y|\int_{0}^{\tau^* }|u(t,\tau)|d\tau \notag \\  
&\quad+\thinspace M_B \int_{0}^{\tau_{\max}(t)}|u(t,\tau)-v(t,\tau)|d\tau + \mu_Me^{-\rho x} \int_{0}^{\tau_{\max}(T)}|u(t,\tau)-v(t,\tau)|d\tau \notag \\ 
&\quad+ \thinspace \mu_M\left(e^{-\rho x} -e^{-\rho y}\right) \int_{0}^{\tau_{\max}(t)}|v(t,\tau)|d\tau  \notag \\ 
\le \thinspace & g R \|u-v\|_1 + g R |x-y| + M_B\|u-v\|_1  \notag \\ 
&\quad+ \mu_M\max\{1,e^{-\rho R}\} \|u-v\|_1  + \mu_M R |x-y| W_{\rho,R}, 
\label{eq:lip-for-f2}
\end{align}
where $ W_{\rho,R}$ is the Lipschitz constant for $e^{-\rho x}$ in $|x|\le R$.
Putting inequalities \eqref{eq:lip-for-f1} and \eqref{eq:lip-for-f2} together, we get
\begin{align}
     \|F(x,u)-F(y,v)\|_{\mathcal{X}} &= |f_1(x,u)-f_1(y,v)| + \|f_2(x,u)-f_2(y,v)\|_1 \notag \\
 & \le    L_{F,R}\|(x,u)-(y,v)\|_{\mathcal{X}}
 \end{align}
 where the Lipschitz constant $L_{F,R}$ is
\begin{align}
  L_{F,R} : = r+M_B + R\left(2a + \max\{s,b\} + g + \mu_M    W_{\rho,R}\right) + \mu_M \max\{1,e^{-\rho R}\}.
\end{align}
This completes the proof of Lemma \ref{lemma:Lip}.
\end{proof}

\subsection{How long does a numerical solution stay in a ball?}
The Lipschitz continuity of $F$, \eqref{eq:Flip}, and the fact that $F(0,0) = (0,0)$ imply that 
\begin{equation}
\label{eq:Fbound}
    \|F(x,u)\|_{\mathcal{X}} = |f_1| + \|f_2\|_1\le L_{F,R}\|(x,u)\|_{\mathcal{X}} \le  L_{F,R}R.
\end{equation}
This allows us to guarantee that a numerical solution with any finite time step exists for at least a minimal time independent of the step size, if the step size is sufficiently small.
\begin{lemma}
    \label{lemma:mintime}
    Let $(X,U)$ be a numerical solution to \eqref{eq:Xnum}--\eqref{eq:Unum_init} with a time step $h$. Let $\|(X,U)\|_{\mathcal{X}}\le R_0$. Then the numerical solution will remain in the ball $\|(X,U)\|_{\mathcal{X}}\le R$ where $R > R_0$ at least for some minimal positive time independent of $h$.
\end{lemma}
\begin{proof}
    For the $x$-component of the solution we have:
    \begin{equation}
    \label{eq:Xk+1}
        |X[k+1]| = |X[k]| + h|f_1(X[k],U[k])| \le 
        |X[k]| + hL_{F,R}R\|(X[k],U[k,\cdot])\|_{\mathcal{X}}.
    \end{equation}
    Next, we estimate $U[k+1,0]$:
    \begin{align*}
        U[k+1,0]& = h{\sf Trap}(B(X[k],\cdot)U[k,\cdot])\le B_{\max}\|U[k,\cdot]\|_1 + O(h)\\
        &\le (B_{\max}+1)\|(X[k],U[k])\|_{\mathcal{X}},
    \end{align*}
    where $h{\sf Trap}(\cdot)$ denotes the trapezoidal rule with step $h$ applied to $(\cdot)$ and $B_{\max}$ is the maximum of $B(X,\tau)$ in the ball $\|(X,U)\|_{\mathcal{X}}\le R$.
    The, we bound $U[k+1,j]$ for $j\ge 1$:
    \begin{align*}
        U[k+1,j] = U[k,j-1] + hf_2(X[k],h(j-1),U[j-1]).
    \end{align*}
    Integration over age, we get
    \begin{align}
    \label{eq:Uk+1}
        \|U[k+1,\cdot]\|_1 \le hC_0\|U[k,\cdot]\|_1 + \|U[k,\cdot]\|_1 + hL_{F,R}R\|(X[k],U[k])\|_{\mathcal{X}}.
    \end{align}
    Adding \eqref{eq:Xk+1} and \eqref{eq:Uk+1} we obtain:
    \begin{align}
        \|(X[k+1],U[k+1])\|_{\mathcal{X}} \le \|(X[k],U[k])\|_{\mathcal{X}}(1 + Ch),
    \end{align}
    where $C: = 2L_{F,R}R + B_{\max}+1$.
Therefore, 
    \begin{align}
        \|(X[k],U[k])\|_{\mathcal{X}} \le \|(X[0],U[0])\|_{\mathcal{X}}(1 + Ch)^k\le R_0e^{Ckh}.
    \end{align}
This means that the numerical solution $(X[k],U[k])$ will remain in the ball $ \|(X[k],U[k])\|_{\mathcal{X}}\le R$ for the time at least
\begin{equation}
    T: = kh = \frac{1}{2L_{F,R}R + B_{\max}+1}\log\frac{R}{R_0} > 0.
\end{equation}
\end{proof}

\subsection{Existence.}
\label{sec:existence}
\begin{lemma}
\label{lem:existence}
Consider a time-space cylinder 
\begin{equation}
    \mathcal{C}:=\{(t,x,u)~|~t\in[0,T],~\|(x,u)\|_{\mathcal{X}}\le R\}.
\end{equation}
Let the initial condition satisfy $\|(x_0,u_0(\tau))\|_{\mathcal{X}} = R_0 < R$. 
Then the initial value problem \eqref{eq:x1}--\eqref{eq:u0} has a solution on the time interval $[0,\min\{T,t_R\}]$ where 
\begin{equation}
    t_R = \inf\{t \ge 0~|~\|(x(t),u(t,\tau))\|_{\mathcal{X}}\ge R\}.
\end{equation}
\end{lemma}

\begin{proof}
We will construct a sequence of approximate, or numerical, solutions to \eqref{eq:x1}--\eqref{eq:u0} and show that this sequence is Cauchy and hence converges. This limit is the desired solution to the initial value problem \eqref{eq:model-prey}--\eqref{eq:smooth_phi}.
A numerical solution is obtained using the discretization \eqref{eq:Xnum}--\eqref{eq:Y2num_init} and time-stepping. 
We fix a large $N$, define $h:=T/N$, and  construct a sequence of numerical solutions with time and age step sizes $2^{-p}h$, where $p = 0,1,2,\ldots$. We denote the numerical solution with the time and age step $2^{-p}h$ by $(X^p,U^p)$.  A continuous numerical solution in the age variable at any fixed time $t = nh$ where $n$ is small enough is defined by linear interpolation of $(X^p[n],U^p[n,\cdot])$ in age. We terminate time stepping as soon as the $\mathcal{X}$-norm of the numerical solution with any time step $2^{-p}h$, $p = 0,1,2,\ldots$, $(X^0[n],U^0[n,\cdot])$, exceeds $R$, or as time reaches $T = Nh$, whichever event happens first. Lemma \ref{lemma:mintime} guarantees that the termination time is bounded from below. Let the terminal time be $T_1 = N_1h$. 

Our goal is to prove that for any $p\in\mathbb{N}$ and any $1\le n\le N_1$,
\begin{equation}
    \left\|(X^0[n],U^0[n,\cdot]) - (X^p[2^pn],U^p[2^pn,\cdot])\right\|_{\mathcal{X}} \le Ah,
\end{equation}
where $A$ is a constant independent of $p$. Note that $2^pn$ time steps of the numerical recurrence \eqref{eq:Xnum}--\eqref{eq:Y2num_init} with steps $h$ and  $h2^{-p}$ yield approximate solutions $(X^0[n],U^0[n,\cdot])$ and $(X^p[2^pn],U^p[2^pn,\cdot])$ to \eqref{eq:x1}--\eqref{eq:u0} at the same time $t = nh$.

{\bf Step 1. The discrepancy between the coarse and fine mesh solutions over the first step.} The first step toward this goal is to show that 
for any $p\in\mathbb{N}$,
\begin{equation}
\label{eq:one_step}
    \left\|(X^0[1],U^0[1,\cdot]) - (X^p[2^p],U^p[2^p,\cdot])\right\|_{\mathcal{X}} \le A_1h^2,
\end{equation}
where $A_1$ is a constant independent of $p$.
The $x$-components of numerical solutions over the time interval $h$ on the coarse $h$, and the fine, $2^{-p}h$, meshes are, respectively,
\begin{align*}
    X^0[1] &= x_0 + hf_1(x_0,u_0(\cdot)),\\
    X^p[1] & = x_0 + 2^{-p}hf_1(x_0,u_0(\cdot)),\\
    &\vdots\\
     X^p[2^p] & = X^p[2^p-1] + 2^{-p}hf_1( X^p[2^p-1],U^p[2^p-1,\cdot]).
\end{align*}
For the $u$-components of these solutions we have:
\begin{align*}
    U^0[1,k] &= u_0(h(k-1)) + hf_2(x_0,u_0(h(k-1)),h(k-1)),\quad k\ge 1,\\
    U^0[1,0] &= h{\sf Trap}(B(x_0,\cdot)u_0(\cdot)),\\
    U^p[1,j] &= u_0(2^ph(j-1)) + 2^{-p}hf_2(x_0,u_0(2^ph(j-1)),2^ph(j-1)),\quad j\ge 1,\\
    U^p[1,0] &= 2^{-p}h{\sf Trap}(B(x_0,\cdot)u_0(\cdot)),\\
    &\vdots\\
    U^p[2^p,j] &= U^p[2^p-1,j-1] + 2^{-p}hf_2(X^p[2^p-1],U^p[2^p-1,j-1],(2^p-1)h(j-1)),\\
    & j\ge 1,\\
    U^p[2^p,0] &= 2^{-p}h{\sf Trap}(B(X^p[2^p-1],\cdot)U^p[2^p-1,\cdot]).
\end{align*}

{\bf The discrepancy in $x$. }Now we bound the difference between $X^0[1]$ and $X^p[2^p]$:
\begin{align}
    |X^p[2^p] - X^0[1]| & \le 2^{-p}h\sum_{l=1}^{2^p-1}| f_1(X^p[l],U^p[l,\cdot]) - f_1(x_0,u_0)|\label{eq:Xstep1}\\
    &\le 2^{-p}h L_{F,R}\sum_{l=1}^{2^p-1} \left(|X^p[l]-x_0| + \| U^p[l,\cdot] - u_0\|_1\right).\notag
\end{align}
The Lipschitz continuity of $F$, \eqref{eq:Flip}, and the fact that $F(0,0) = (0,0)$ imply that 
\begin{equation}
\label{eq:Fbound}
    \|F(x,u)\|_{\mathcal{X}} = |f_1| + \|f_2\|_1\le L_{F,R}\|(x,u)\|_{\mathcal{X}} \le  L_{F,R}R.
\end{equation}
The bound \eqref{eq:Fbound} allows us to bound $|X^p[l] - x_0|$, $0 \le l \le 2^{p}-1$, as follows:
\begin{align}
    |X^p[l]-x_0| = \left|x_0 + 2^{-p}h\sum_{i = 0}^{l-1} f_1(X^p[i],U^p[i,\cdot]) - x_0 \right| \le 2^{-p}hlL_{F,R}R. \label{eq:step1Xbound}
\end{align}
Equation \eqref{eq:Fbound} also allows us to bound $\| U^p[l,\cdot] - u_0\|_1$. Since $u_0$ is piecewise smooth,
\begin{align}
   \| U^p[l,\cdot] - u_0\|_1 = 2^{-p}h{\sf Trap}(U^p[l,\cdot] - u_0 ) + O(h).
\end{align}
Therefore, we need to bound the trapezoidal rule sum. We continue:
\begin{align}
    2^{-p}h{\sf Trap}(|U^p[l,\cdot] - u_0| )
   & =2^{-p}h{\sf Trap}\left|u_0(\cdot) + 2^{-p}h\sum_{i=1}^l G(\cdot) - u_0(\cdot) \right| \notag\\
   & = 2^{-p}h{\sf Trap}\left|2^{-p}h\sum_{i=i_0}^{l-1} G(\cdot) \right|\label{eq:Utrap}\\
   & = 2^{-p}h\sum_{i=i_0}^{l-1}2^{-p}h{\sf Trap}|G(\cdot)|,\notag
\end{align}
where $i_0$ is $0$ if the point index $j \ge l$ and $i_0 = l-j$ if $0\le j< l$, and $G(\cdot)$ stands for $f_2(\cdot)$ or $2^{-p}h{\sf Trap}(B(\cdot)U^p[\cdot])$ -- see Fig. \ref{fig:ExistenceProofFig}. 
\begin{figure}[h!] 
    \centering    
    \includegraphics[width=\textwidth]{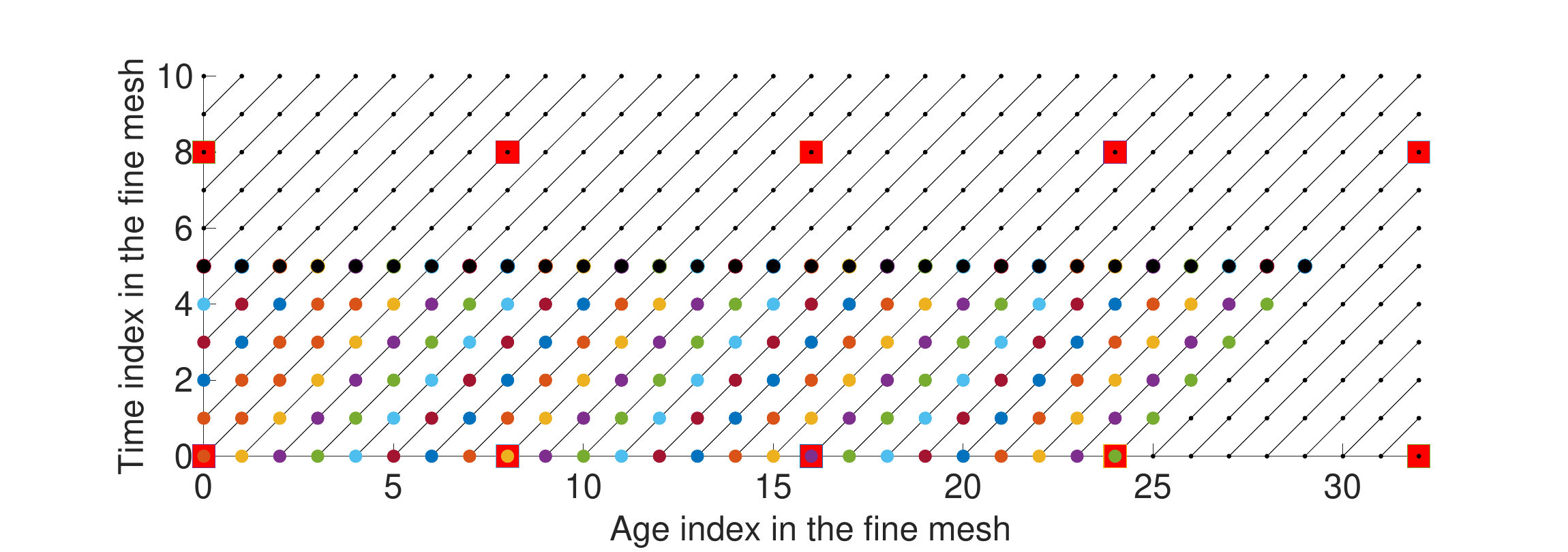}
    \caption{An illustration to the proof of Lemma \ref{lem:existence}, particularly to the calculation in \eqref{eq:Utrap}. The large red squares mark the mesh with step $h$ while the small black dots represent the mesh with step $2^{-p}h$, where $p = 3$. Here, $l= 5$ in \eqref{eq:Utrap}. The large black dots mark the row with $l= 5$. The summation in $i$ of $G(\cdot)$ is done over the points of the same color arranged along the characteristics drawn with black lines. Al all colored dots lying on the $y$-axis, $G(\cdot)$ is  $2^{-p}h{\sf Trap}(B(\cdot)U^p[\cdot])$, while at all other colored dots, $G(\cdot)$ is $f_2(\cdot)$.}
    \label{fig:ExistenceProofFig}
\end{figure}

Since $\|f_2\|_1\le L_{F,R}R$, $B(\cdot)\le B_{\max}$, and $\|u_2\|_1\le R$, we have
\begin{align}
\label{eq:step1Ubound}
     \| U^p[l,\cdot] - u_0\|_1 \le h2^{-p}l\left(\max\{B_{\max},L_{F,R}\}R + O(h)\right)\le A_uh2^{-p}l
\end{align}
for $h\le 1$ where $A_u : = \max\{B_{\max},L_{F,R}\}R + 1$.
Plugging \eqref{eq:step1Xbound} and \eqref{eq:step1Ubound} into \eqref{eq:Xstep1} and computing the sum of the arithmetic series,  we obtain the desired bound for the discrepancy in $x$ between the coarse and fine mesh solutions over the first step in $h$:
\begin{align}
      |X^p[2^p] - X^0[1]| 
     \le\thinspace & h^2L_{F,R}2^p(2^p-1)2^{-2p}\left[L_{F,R}R + A_u \right] \le A_{1,x} h^2. \label{eq:Xerror_step1} 
\end{align}

{\bf The discrepancy in $u$.} Next, we need to show that
\begin{align}
\label{eq:U1bound}
    \left\|U^p[2^p,\cdot] - U^0[1,\cdot]\right\|_1 \le A_{1,u}h^2.
\end{align}
The values of $U^p$ and $U^0$ between the fine and coarse mesh points, respectively, are defined by linear interpolation. Therefore, the $L_1$-norm is exactly equal to the quadrature by the trapezoidal rule on the fine mesh with step $2^{-p}h$. We will the trapezoidal rule quadrature with step $h$ and then estimate the correction due to $U^p$.
We start with the discrepancy at $\tau = 0$:
\begin{align*}
  &  \left|U^p[2^p,0] - U^0[1,0]\right|  = \\
 & \left|2^{-p}h{\sf Trap}\left(B(X^p[2^p-1],\cdot)U^p[2^p-1,\cdot]\right) - h{\sf Trap}\left(B(x_0,\cdot)u_0(\cdot)\right)\right| \le \\
    & \left|2^{-p}h{\sf Trap}\left(B(X^p[2^p-1],\cdot)U^p[2^p-1,\cdot]\right) - 2^{-p}h{\sf Trap}\left(B(x_0,\cdot)u_0(\cdot)\right)\right|  +\\
    &\left|2^{-p}h{\sf Trap}\left(B(x_0,\cdot)u_0(\cdot)\right) - h{\sf Trap}\left(B(x_0,\cdot)u_0(\cdot)\right)\right| \le \\
   & \left|2^{-p}h{\sf Trap}\left(B(X^p[2^p-1],\cdot)U^p[2^p-1,\cdot]\right) - 2^{-p}h{\sf Trap}\left(B(x_0,\cdot)U^p[2^p-1,\cdot]\right)\right|  + \\
   &\left|2^{-p}h{\sf Trap}\left(B(x_0,\cdot)U^p[2^p-1,\cdot]\right) - 2^{-p}h{\sf Trap}\left(B(x_0,\cdot)u_0(\cdot)\right)\right|  +\\ 
  & \left|2^{-p}h{\sf Trap}\left(B(x_0,\cdot)u_0(\cdot)\right) - h{\sf Trap}\left(B(x_0,\cdot)u_0(\cdot)\right)\right|
\end{align*}
To continue, we take into account the following facts:
\begin{itemize}
    \item $B$ is Lipschitz in $x$ with constant $L_B$ and $|X^p[2^p-1] - x_0|\le hL_{F,R}R$ according to \eqref{eq:step1Xbound};
    \item $B\le B_{\max}$ and $ \| U^p[2^p-1,\cdot] - u_0\|_1\le A_uh$ according to \eqref{eq:step1Ubound};
    \item since $B$ and $u_0$ are piecewise smooth, the difference between the trapezoidal rule quadrature with steps $h$ and $2^{-p}h$ is $O(h)$.
\end{itemize}
Therefore, we conclude that if $h$ is small enough,
\begin{align}
     \left|U^p[2^p,0] - U^0[1,0]\right| &\le L_BL_{F,R}Rh\|U^p[2^p-1,\cdot]\|_1 + B_{\max}A_uh + O(h)\notag \\
     &\le A_{1,0}h. \label{eq:U0bound}
\end{align}
Next, we calculate the discrepancies at a coarse mesh point with index $[1,k],$ with $k > 0$. This corresponds to the fine mesh index $[2^p,2^pk]$:
\begin{align*}
    &\left|U^p[2^p,2^pk]-U^0[1,k]\right| = \left|u_0 + 2^{-p}h\sum_{i=0}^{l-1}f_2(\cdot) - u_0 -hf_2(x_0,u_0(h(k-1))) \right|\\
    \le\thinspace & 2^{-p}h\sum_{i=0}^{2^p-1}|f_2[p,i,2^p(k-1)+i] - f_2(x_0,h(k-1),u_0(h(k-1))|
    % \le \thinspace & 2^{-p}h\sum_{i=0}^{2^p-1}|U^p[i,2^p(k-1)+i]\mu(X^p[i],2^{-p}h(2^p(k-1)+i)) - u_0\mu(x_0,h(k-1))|,
\end{align*}
where we used a  short-cut notation 
\begin{equation}
    f_2[p,i,2^p(k-1)+i]:=f_2(X^p[i],2^{-p}h(2^p(k-1)+i),U^p[i,2^p(k-1)+i]).
\end{equation}
Therefore,
\begin{align*}
   & h{\sf Trap}(|U^p[2^p,\cdot] - U^0[1,\cdot]|) \le \frac{h}{2}A_{1,0}h + h\sum'_{k\ge 1}\left|U^p[2^p,2^pk]-U^0[1,k]\right| \\
   \le \thinspace & 
   \frac{h^2}{2}A_{1,0} +  2^{-p}h\sum_{i=0}^{2^p-1}\left(L_{F,R}\|(X^p[2^p],U^p[2^p,\cdot]) - (x_0,u_0)\|_{\mathcal{X}} + O(h)\right) \le A'_1h^2,
\end{align*}
where the sum with a prime above it indicates that the last term in it is divided by 2. 
This trapezoidal rule is exact for $U^0$ but is not exact for $U^p$. However, on all intervals $[kh,(k+1)h]$ with $k\ge 1$, $U^p[2^p,\cdot]$ can be extended to a piecewise smooth function thanks to such property of $u_0$ and $f_2$. Therefore, the trapezoidal rule error on in most of these intervals will be $O(h^3)$ and in a few corresponding to jump discontinuities, $O(h^2)$. The error on the first interval, $[0,h]$, will be $O(h^2)$, because there $U^p$ can be extended continuously to $[0,h)$, but may have a discontinuity at $h$. Therefore, the overall error of the trapezoidal rule will be $O(h^2)$.
Therefore, we conclude that \eqref{eq:U1bound} holds. Moreover, $U^p[2^p,\cdot]$ can be extended into a piecewise smooth function.

{\bf Step 2. The accumulation of the discrepancy between the solutions of coarse and fine meshes over a finite time interval.} We will introduce the following {notations}
\begin{equation}
    e_{k,x}: = X^p[2^{p}k] - X^0[k],\quad e_{k,u}[\cdot]: = U^p[2^pk,\cdot]-U^0[k,\cdot].
\end{equation}
We also will use a short-cut notation 
\begin{equation}
   f_1[p,2^pk]:= f_1(X^p[2^pk],U^p[2^pk,\cdot]).
\end{equation}
First, we compute $e_{k+1,x}$ given $e_{k,x}$:
\begin{align*}
    e_{k+1,x}& : = X^p[2^{p}(k+1)] - X^0[k+1] \\
    & = X^p[2^{p}k] - X^0[k] +2^{-p}h\sum_{l=0}^{2^p-1}\left(f_1[p,2^pk+l] - f_1[0,k] \right)\\
    & = e_{k,x} + 2^{-p}h\sum_{l=0}^{2^p-1}\left(f_1[p,2^pk+l] - f_1[p,2^pk]+f_1[p,2^pk] - f_1[0,k]\right).
\end{align*}
Using the Lipschitz continuity of $f_1$ and \eqref{eq:one_step} we obtain
\begin{align}
\label{eq:ek+1x}
    |e_{k+1,x}| \le |e_{k,x}| + h\left(L_{F,R} |e_{k,x}| + A_1 h^2\right).
\end{align}

Second, we compute $e_{k+1,u}(0)$ given $e_{k,u}(\cdot)$:
\begin{align*}
    e_{k+1,u}[0]&: = 2^{-p}h{\sf Trap}(B[p,2^pk]U^p[2^pk,\cdot]) - h {\sf Trap}(B[0,0k]U^0[k,\cdot]) \\
    & = 2^{-p}h{\sf Trap}(B[p,2^pk]U^p[2^pk,\cdot]) - 2^{-p}h{\sf Trap}(B[0,k]U^p[2^pk,\cdot])\\
    &\quad+\thinspace2^{-p}h{\sf Trap}(B[0,k]U^p[2^pk,\cdot]) - 2^{-p}h{\sf Trap}(B[0,k]U^0[k,\cdot])\\
    &\quad+\thinspace  2^{-p}h{\sf Trap}(B[0,k]U^0[k,\cdot]) - h {\sf Trap}(B[0,k]U^0[k,\cdot]),
\end{align*}
where $B[p,2^pk]:= B(X^p[2^pk],U[2^pk,\cdot])$.
Taking into account the facts listed in the bullet list above \eqref{eq:U0bound} we obtain
\begin{align}
   | e_{k+1,u}[0]| \le L_B\left(|e_{k,x}| + \|e_{k,u}\|_1\right) + B_{\max}\|e_{k,u}\|_1 + O(h).
\end{align}
Third, we compute $e_{k+1,u}(j)$ for $j\le 1$ referring to the coarse mesh in age. We will use a short-cut notation
\begin{align}
    f_2[p,2^pk,2^pj]:= f_2(X^p[2^pk],2^{-p}h2^pk,U^p[2^pk,2^pj]).   
\end{align}
Thus, we calculate:
\begin{align*}
    &e_{k+1,u}[j]: = U^p[2^p(k+1),2^pj] - U^0[k+1,j] \\
   =\thinspace &  U^p[2^pk,2^p(j-1)] + 2^{-p}h\sum_{i=0}^{2^p-1}f_2[p,2^pk+i,2^p(j-1)+i]\\
  &\quad -\thinspace U^0[k,j-1] - hf_2[0,k,j-1]\\
   =\thinspace &  e_{k,u}[j-1] + 2^{-p}h\sum_{i=0}^{2^p-1}\left(f_2[p,2^pk+i,2^p(j-1)+i] - f_2[p,2^pk,2^p(j-1)]\right)\\
  &\quad +\thinspace   2^{-p}h\sum_{i=0}^{2^p-1}\left(f_2[p,2^pk,2^p(j-1)] - f_2[0,k,j-1]\right).
\end{align*}
Summing $e_{k,u}$ over the age range using the trapezoidal rule and taking into account that the numerical solutions can be extended to piecewise smooth functions, we get
\begin{align}
\label{eq:ek+1u}
    \|e_{k+1,u}\|_1 \le \frac{h}{2}L_B\|e_k\|_{\mathcal{X}} + \left(\frac{h}{2}B_{\max}+1\right)\|e_{k,u}\|_1  + hL_{F,R}\|e_{k}\|_{\mathcal{X}} + O(h^2),
\end{align}
where $\|e_k\|_{\mathcal{X}}: = |e_{k,x}| + \|e_{k,u}\|_1$.
Adding \eqref{eq:ek+1x} and \eqref{eq:ek+1u} we obtain
\begin{align}
    \|e_{k+1}\|_{\mathcal{X}} \le  \|e_{k}\|_{\mathcal{X}}(1 + C_1h) + C_2h^2,
\end{align}
where $C_1$ and $C_2$ are appropriate positive constants.
This implies that 
\begin{align}
    \|e_{k}\|_{\mathcal{X}} &\le (1 + C_1h)^k \|e_{0}\|_{\mathcal{X}} + \frac{(1 + C_1h)^k-1}{1 + C_1h -1}C_2h^2 \notag \\
    &= 
    e^{C_1hk}  \|e_{0}\|_{\mathcal{X}}+ \left(e^{C_1hk}-1\right)\frac{C_2}{C_1}h.
\end{align}
Since $hk\le hN_1 = T_1$ and $\|e_{0}\|_{\mathcal{X}} = 0$, we obtain the desired bound:
\begin{align}
\label{eq:ekX}
    \|e_{k}\|_{\mathcal{X}}\le \left(e^{C_1T_1} -1\right)\frac{C_2}{C_1}h.
\end{align}
\end{proof}

{\bf Step 3. Taking the limit.}
The bound \eqref{eq:ekX} shows that the sequence $(X^p,U^p)$ of numerical solutions at any fixed time $t < T_1$ with steps $2^{-p}h$ is Cauchy in the Banach space $\mathcal{X}$ and hence converges to an element $(x,u)$ in $\mathcal{X}$. The $x$-component and $u$ along the characteristics are continuously differentiable as immediately follows from time stepping. The function $u(t,\tau)$ may have discontinuities propagating characteristics emanating from the origin of the $(t,\tau)$-plane and from the finite number of discontinuities of the initial age density $u_0$ -- see the remark in Section \ref{sec:main_theorem}.

\subsection{Uniqueness}
\begin{lemma}
    \label{lemma:unique}
    Under conditions of Theorem \ref{thm1}, the solution to \eqref{eq:x1}--\eqref{eq:u0} is unique.
\end{lemma}
\begin{proof}
Let $(x,u)$ and $(\tilde{x},\tilde{u})$ be two solutions of \eqref{eq:x1}--\eqref{eq:u0} on the time interval $[0,T_1]$ defined in Section \ref{sec:existence}.
We define a vector $y(t)$ as follows:
\begin{align*}
    y(t): = \left(\begin{array}{c} x(t)-\tilde{x}(t)\\
    u(t,0) - \tilde{u}(t,0)\\
    u(t,t-t_0)-\tilde{u}(t,t-t_0)\\
    u(t,t+\tau_0)-\tilde{u}(t,t+\tau_0)\end{array}\right),
\end{align*}
where $t_0> 0$ is the birth time, $\tau_0> 0$ is the initial age. We separate the age-zero term because it generates the initial conditions for all characteristics with $t_0 \ge 0$ emanating from the $t$-axis.
We introduce the norm
\begin{align}
\label{eq:ynorm}
      \|y(t)\|_{\mathcal{Y}} = |x(t)-\tilde{x}(t)| + |u(t,0) - \tilde{u}(t,0)| 
      + \int_0^{\infty}|u(t,\tau)-\tilde{u}(t,\tau)|d\tau.
\end{align}
and a space $\mathcal{Y}$ consisting of all elements $(z,v_1(0),v_1(\cdot),v_2(\cdot))$ with finite $\mathcal{Y}$-norm: $|z| + |v_1(0)|+\|v_1\|_1 + \|v_2\|_1 < \infty$. The $\mathcal{Y}$-space is Banach as it is a direct product of a Banach space and an intersection of Banach spaces. As in the $\mathcal{X}$-space, we consider a ball of a large radius $R$ in the $\mathcal{Y}$-space and assume that the solutions $(x,u)$ and $(\tilde{x},\tilde{u})$ lie in this ball. 
Next, we consider a Fréchet derivative of $y(t)$:
\begin{align}
    y'(t) = \left(\begin{array}{c} f_1(x(t),u(t,\cdot))-f_1(\tilde{x}(t),\tilde{u}(t,\cdot))\\
   \int_0^{\infty} B(x(t),\tau)u_t(t,\tau)d\tau -  \int_0^{\infty} B(\tilde{x}(t),\tau)\tilde{u}_t(t,\tau)d\tau\\
    f_2(x(t),t-t_0,u(t,t-t_0))-f_2(\tilde{x},t-t_0,(t),\tilde{u}(t,t-t_0))\\
     f_2(x(t),t+\tau_0,u(t,t-t_0))-f_2(\tilde{x},t+\tau_0,(t),\tilde{u}(t,t+\tau_0))\end{array}\right).    
\end{align}
Using  the PDE for the predator age density \eqref{eq:model-predator}, we rewrite integrals in the second component of $y'(t)$ as follows:
\begin{align*}
   & \int_0^{\infty} B(x(t),\tau)u_t(t,\tau)d\tau =  \\
    -&\int_0^{\infty} B(x(t),\tau)u_\tau(t,\tau)d\tau - \int_0^{\infty} B(x(t),\tau)\mu(x(t),\tau)u(t,\tau)d\tau =
    \\
    -&B(x(t),0)u(t,0) + \int_0^{\infty} B_{\tau}(x(t),\tau)u(t,\tau)d\tau -\thinspace \int_0^{\infty} B(x(t),\tau)\mu(x(t),\tau)u(t,\tau)d\tau.
\end{align*}
 Taking into account that $u(t,\tau)$ has a compact support in $\tau$ at any $t<\infty$ as $u_0(\tau)$ is compactly supported, and that $(x,u)$ and $(\tilde{x},\tilde{u})$ lie in the ball of radius $R$ in the $\mathcal{Y}$-space, we obtain  the following bound for the second component of $y'(t)$
\begin{align*}
    & \left| \int_0^{\infty} B(x,\tau)u_t(t,\tau)d\tau -  \int_0^{\infty} B(\tilde{x}(t),\tau)\tilde{u}_t(t,\tau)d\tau\right|\\
    \le\thinspace &|B(x,0)||u(t,0)-\tilde{u}(t,0)| + |B(x,0)-B(\tilde{x},0)||\tilde{u}(t,0)|\\
   &\quad +\thinspace
    \int_0^{\infty}|B_{\tau}(x,\tau)||u(t,\tau)-\tilde{u}(t,\tau)|d\tau 
    +\int_0^{\infty}|B_{\tau}(x,\tau)- B_{\tau}(\tilde{x},\tau)||\tilde{u}(t,\tau)|d\tau \\
    &\quad +\thinspace \int_0^{\infty}|B(x,\tau)\mu(x,\tau)||u(t,\tau)-\tilde{u}(t,\tau)|d\tau\\ 
    &\quad +\thinspace \int_0^{\infty}|B(x,\tau)\mu(x,\tau)-B(\tilde{x},\tau)\mu(\tilde{x},\tau)||\tilde{u}(t,\tau)|d\tau  \\
    \le\thinspace &
    |B(x,0)||u(t,0)-\tilde{u}(t,0)| + L_{B_{\tau}}|x-\tilde{x}|R\\
    &\quad +\thinspace B_{\max}\mu_{\max}\|u(t,\cdot)-\tilde{u}(t,\cdot)\|_1 + L_{B\mu}|x-\tilde{x}|R,
\end{align*}
    where $L_{B_{\tau}}$ and $L_{B\mu}$ are the Lipschitz constants for $B_{\tau}$ and $B\mu$ respectively. Note that $B_{\tau}(x,\tau)$ is Lipschitz in $x$ because the smoothed indicator function has finite derivative  at all $\tau$.
Taking Lipschitz continuity of $f_1$ and $f_2$ into account and using the bound for the second component of $y'(t)$ just derived, we claim that
\begin{equation}
\label{eq:yFrechet}
    \|y'(t)\|_{\mathcal{Y}}\le L\|y(t)\|_{\mathcal{Y}}
\end{equation}
for an appropriate Lipschitz constant $L$.

Next, we fix a small $h>0$. By the triangle inequality,
\begin{equation*}
    \|y(t+h)\|_{\mathcal{Y}} \le  \|y(t+h) - y(t)\|_{\mathcal{Y}} + \|y(t)\|_{\mathcal{Y}}.
\end{equation*}
Hence,
\begin{equation*}
    \frac{\|y(t+h)\|_{\mathcal{Y}} -  \|y(t)\|_{\mathcal{Y}}}{h} \le  \frac{\|y(t+h) - y(t)\|_{\mathcal{Y}}}{h}.
\end{equation*}
The same is true if $h <0$ and $t > -h$.
Together with \eqref{eq:yFrechet} this implies that
\begin{align}
    \frac{d}{dt}\|y\|_{\mathcal{Y}}\le \|y'(t)\| \le L\|y\|_{\mathcal{Y}}.
\end{align}
 Therefore, by Gr\"{o}nwall's Inequality,
 \begin{equation}
     0\le \|y(t)\|_{\mathcal{Y}}\le \|y(0)\|e^{Lt}.
 \end{equation}
 Since $y(0) = 0$, it remains zero for all times. Hence the solution to \eqref{eq:x1}--\eqref{eq:u0} is unique.
 \end{proof}

\subsection{Positivity}
{\bf Positivity of $x(t)$.}
\begin{lemma}
\label{lemma:Xpos}
    Under the conditions of Theorem \ref{thm1}, the $x$-component of the solution is positive provided that $x(0)$ is positive.
\end{lemma}
\begin{proof}
    Equation \eqref{eq:x1} for $x$ can be rewritten as follows:
 \begin{align}
    x'(t) = x(t)f(x(t),u(t,\cdot)),
\end{align}
where 
\begin{align}
    f(t) = r - ax(t) + s\int_0^{\tau^*} u(t,\tau)d\tau - b\int_{\tau^*}^{\tau_{\text{max}}(T)} u(t,\tau) d\tau.
\end{align}
Therefore, we can introduce $z(t):=\log x(t)$ and write the following ODE for $z$:
\begin{align}
    z'= f(e^z,u(t,\cdot)).
\end{align}
 If we replace $x$ with $z$ in \eqref{eq:x1}--\eqref{eq:u0}, we can discretize it and conduct the existence and uniqueness proofs essentially repeating the arguments in the previous sections of Appendix A. Importantly, the initial condition for $z$, $z(0) = \log x(0)$ is defined. Hence, a solution $(z,u)$ exists. Then $x(t) = e^{z(t)}$ also exists and is positive.
\end{proof}

{\bf Positivity of $u(t,\tau)$.}
\begin{lemma}
Let $u_0(\tau)$ be positive on $[0,\tau_{0,\max})$.
    Under the conditions of Theorem \ref{thm1}, the $u$-component of the solution is positive  in the trapezoidal $(t,\tau)$-region
    \begin{equation}
        \mathcal{T}: = \{(t,\tau)~|~0\le t\le T,~0\le\tau < \tau_{0,\max} + t\}.
    \end{equation}
\end{lemma}
\begin{proof}
    One can argue that $u(t,t+\tau_0) > 0$ and $u(t,t-t_0)>0$ as soon as $u_0(\tau_0) >0$ and $u(t_0,0)>0$ in the same manner as in the proof of Lemma \ref{lemma:Xpos}. It is given that $u_0(\tau) > 0$ on $[0,\tau_{0,\max})$. Therefore, the only opportunity for $u$ to become nonpositive in $\mathcal{T}$ is to acquire a nonpositive initial condition at $(t_0,0)$ for some $t_0 \ge 0$. This means that the birth integral must become nonpositive. Let $t_0 \ge 0$ be the smallest time at which
    \begin{align*}
        u(t_0,0) = \int_0^{\infty}B(x(t_0),\tau)u(t_0,\tau)d\tau\le 0.
    \end{align*}
    This mean that $u(t_0,\tau)\le 0$ at some set of $\tau$ of positive measure inside  $[0,\tau_{0,\max} + t_0)$. Let $\tau_1\in [0,\tau_{0,\max} + t_0)$ be such that $u(t_0,\tau_1) \le 0$. Then, since $u$ preserves its sign along the characteristics, $u$ must be negative along the whole characteristic passing through $(t_0,\tau_1)$. If this characteristic hits the $\tau$-axis first, we come to a contradiction with the assumption that $u_0 > 0$ on $[0,\tau_{0,\max})$. If this characteristic hits the $t$-axis first, we arrive at a contradiction with the assumption that $t_0$ is the earliest moment of time when $u(t,\tau)$ became nonpositive on $\tau\in [0,\tau_{0,\max} + t)$. Hence, $u(t,\tau)$ is positive in $\mathcal{T}$.
\end{proof}

\section{Linear Discriminant Analysis}
\label{sec:AppB}
Let $X$ be a $n\times d$ matrix whose rows are $d$-dimensional data points. The dataset $X$ consists of  $c>1$ categories. 
Let $\mathcal{I}_i$ denote the set of indices of data from category $i$, $i=1,\ldots,c$, and
\begin{equation*}
|\mathcal{I}_i| = n_i,\quad \mathcal{I}_i\cap \mathcal{I}_j =\emptyset,\quad\mathcal{I}_1\cup\ldots\cup \mathcal{I}_c = \{1,\ldots,n\}.
\end{equation*}

We seek a $d\times d_1$ matrix $W$, $d_1\le c-1$, that maps the data onto a $d_1$-dimensional space:
\begin{equation}
\label{e1}
Y = XW,\quad{\rm or},\quad y_k = W^\top x_k,\quad k=1,\ldots n,
\end{equation}
such that images of the data from different categories under this mapping are separated as much as possible while data from the same categories are clustered as much as possible.
In \eqref{e1}, $x_i\in\mathbb{R}^d$ is the $k$th row of $X$ written as a column vector. Likewise is $y_k$.

The means for each category in the spaces $\mathbb{R}^d$ and $\mathbb{R}^{d_1}$ are, respectively,
\begin{equation}
m_i = \frac{1}{n_i}\sum_{k\in\mathcal{I}_i}x_k,\quad {\rm and}\quad\tilde{m}_i = \frac{1}{n_i}\sum_{k\in\mathcal{I}_i}y_k = W^\top m_i,\quad i = 1,\ldots, c.
\end{equation}

To describe the data variation within and between the categories, the \emph{within-class} and \emph{between-class scatter matrices} $S_w$ and $S_b$ are introduced.
These matrices are of size $d\times d$.
The within-class scatter matrix is defined as
\begin{equation}
S_w : = \sum_{i=1}^c S_i,
\end{equation}
where $S_i$ is the scatter matrix of category $i$ defined as
\begin{equation}
S_i:= \sum_{k\in\mathcal{I}_i}(x_k-m_i)(x_k-m_i)^\top\equiv \left(X_{\mathcal{I}_i,:} - \mathbf{1}_{n_i\times 1}m_i^\top\right)^\top \left(X_{\mathcal{I}_i,:} - \mathbf{1}_{n_i\times 1}m_i^\top\right).
\end{equation}
The between-class scatter matrix is defined as
\begin{equation}
S_b: = \sum_{i=1}^c n_i(m_i-m)(m_i-m)^\top,\quad{\rm where}\quad 
m: = \frac{1}{n}\sum_{i=1}^cn_im_i \equiv \frac{1}{n}\sum_{k=1}^nx_k
\end{equation}
is the overall mean. The rank of $S_b$ is at most $c-1$.

A similar notation with tilde on top is used for the mapped data. The within- and between-class scatter matrices for the mapped data are:
\begin{equation}
\label{ySwSb}
\tilde{S}_w = W^\top S_w W,\quad \tilde{S}_b =  W^\top S_b W.
\end{equation}

The task of LDA is to find a $d\times d_1$ matrix $W$, $d_1\le c-1$, mapping the data within each category into clusters and mapping the clusters corresponding to different categories as far as possible from each other. Respectively, the LDA objective function is defined as
\begin{equation}
J(w) = \frac{w^\top S_b w}{w^\top S_w w}.
\end{equation}
The function $J(w)$ is maximized by solving the generalized eigenvalue problem. Indeed, the gradient of $J(w)$ is zero if and only if 
\begin{equation}
\label{eq:geneig}
S_bw = J(w)S_ww.
\end{equation}
Therefore, to find the desired projection, one needs to solve \eqref{eq:geneig} and compose the projection matrix $W$ out of $d_1$ eigenvectors corresponding to $d_1$ largest eigenvalues.

%%%%%%%%%%%%%%%%
\bibliographystyle{unsrtnat}
%\bibliography{reference}

\end{document}